\documentclass{article}

\PassOptionsToPackage{numbers, compress}{natbib}
\usepackage[preprint]{neurips_2026}

\usepackage[utf8]{inputenc} % allow utf-8 input
\usepackage[T1]{fontenc}    % use 8-bit T1 fonts
\usepackage{hyperref}       % hyperlinks
\usepackage{url}            % simple URL typesetting
\usepackage{booktabs}       % professional-quality tables
\usepackage{amsfonts}       % blackboard math symbols
\usepackage{nicefrac}       % compact symbols for 1/2, etc.
\usepackage{microtype}      % microtypography
\usepackage{xcolor}         % colors

\usepackage{etoolbox}
\usepackage{array}
\usepackage{pifont}
\usepackage{comment}

\usepackage{amsmath}
\usepackage{amssymb}
\usepackage{mathtools}
\usepackage{amsthm}

\usepackage{graphicx}
\usepackage{subfigure} 
\usepackage{subcaption}
\usepackage{multirow}
\usepackage{makecell}
\usepackage{diagbox}
\usepackage[export]{adjustbox}

\usepackage{algorithm}
\usepackage[noend]{algorithmic}

\usepackage{soul}
\soulregister\cite7
\soulregister\ref7
\soulregister\pageref7

\definecolor{DeepGreen}{RGB}{0,128,0}
\sethlcolor{green!30}

\newcommand{\hlred}[1]{{\sethlcolor{red!25}\hl{#1}}}
\newcommand{\hlyellow}[1]{{\sethlcolor{yellow!50}\hl{#1}}}
\newcommand{\hlg}[1]{{\sethlcolor{green!100}\hl{#1}}}

\newcommand\eat[1]{}

\usepackage{varwidth}
\usepackage{tcolorbox}
\tcbuselibrary{breakable}

\theoremstyle{plain}
\newtheorem{theorem}{Theorem}[section]

\theoremstyle{definition}

\theoremstyle{remark}

\title{Robust LLM Watermarking with \\ Minimal Semantic Distortion for IP Protection}

\author{%
  Kieu Dang$^{1}$ \quad
  Phung Lai$^{1}$ \quad
  NhatHai Phan$^{2}$ \quad
  Yelong Shen$^{3}$ \quad
  Ruoming Jin$^{4}$ \\
  $^{1}$State University of New York at Albany \\
  $^{2}$New Jersey Institute of Technology \\
  $^{3}$Microsoft \\
  $^{4}$Kent State University \\
  \texttt{\{vdang,lai\}@albany.edu} \quad
  \texttt{phan@njit.edu} \quad
  \texttt{Yelong.Shen@microsoft.com} \quad
  \texttt{rjin1@kent.edu}
}

\begin{document}

\maketitle

\begin{abstract}
 Proprietary large language models (LLMs) face risks of intellectual property (IP) violation, as adversaries can replicate an LLM by collecting input-output pairs to train a surrogate model, causing  financial setbacks. Watermarks offer a promising defense to verify ownership, but existing methods often struggle with semantic distortion, factual inconsistency, and adversarial attacks.  
 In addition, key-conditioned watermarks for provider-specific detection, especially in cross-provider and multi-user scenarios, remain largely underexplored. To address these challenges, we propose \textsc{SafeSeal}, a novel key-conditioned watermarking framework that achieves strong detectability with minimal impact on model utility, effectively balancing detectability, utility, and robustness. \textsc{SafeSeal} preserves named entities while substituting   linguistic terms with context-aware synonyms through a key-conditioned Tournament sampling mechanism,  maintaining semantic fidelity and factual consistency. For detection, we introduce a key-conditioned contrastive detector that jointly encodes the text and key, enabling provider-specific and robust watermark verification. We derive theoretical bounds on the utility-detectability trade-off and significantly reduce latency through lightweight models, batching, and parallelism. Extensive experiments show that \textsc{SafeSeal} outperforms baselines in utility,  detectability, and robustness, achieving a BERTScore of $0.983$, entity similarity of $0.963$, a $98.2\%$ detection rate, and the highest human ratings for text quality and content preservation, with latency comparable to the fastest baseline. To promote transparency and community-driven progress, we release the first public watermark leaderboard and an interactive demo.

 %, particularly when susceptible content, i.e., named entities, is altered. 
  %By constructing a lookup table of linguistic terms and their synonyms with quantifiable similarity, we derive theoretical bounds that shed light on the correlation between content deviation and   detectability.  
 % In addition, \textsc{SafeSeal} improves 
 %in attack-free settings and under removal and model stealing attacks. 
 % It achieves high semantic similarity (a BERTScore of $0.983$) and entity similarity ($0.963$) with a detection rate of $98.20\%$. \textsc{SafeSeal} also achieves the highest human ratings for text quality and content preservation, with latency comparable to existing watermarks. To promote transparency and community-driven progress, we release the first public watermark leaderboard and an interactive demo for watermark testing.

\end{abstract}

\section{Introduction}
Large language models (LLMs), such as ChatGPT \cite{OpenAI},   DeepSeek \cite{liu2024deepseek}, and Gemini \cite{gemini}, have shown powerful capabilities in text generation,   translation, and comprehension.  Although LLMs are often deployed as services via APIs that limit access to model weights \cite{minthigpen,ibm}, they remain vulnerable to model stealing attacks \cite{zhang2023apmsa,zeng2024huref,xu2022student,krishna2019thieves,juuti2019prada}, where adversaries fine-tune surrogate models on proprietary outputs, threatening providers' intellectual property (IP)  \cite{li2023protecting}.
 LLM watermarks \cite{kirchenbauer2023watermark,yoo2023advancing,liu2023unforgeable,xu2023watermarking,jia2021entangled,qu2025provably,chen2025improved,ren2023robust,zhu2024duwak,liang2024watme,fu2024gumbelsoft}  address the risk by embedding imperceptible patterns   into LLM outputs, enabling service providers to detect unauthorized use of their LLM-generated text. 

However, existing LLM watermarks \cite{yoo2024advancing,lee2023wrote,meral2009natural} often cause substantial content deviation and  utility  loss to gain strong   detectability, particularly under attacks  \cite{molenda2024waterjudge,sadasivan2023can,piet2025mark,pang2024no,ajith2024downstream}. An effective watermark should preserve semantics and factual consistency, offer high detectability, resist attacks, and maintain low latency.   However, these goals are difficult to balance: weak signals reduce detectability, while strong signals degrade utility and may increase computational overhead.
%Achieving these goals simultaneously is fundamentally challenging, as weak watermark signals (i.e., subtle text edits) can reduce watermark detectability, while strong watermark signals degrade text quality and model  utility \cite{molenda2024waterjudge}. In addition, strengthening the watermark may enhance robustness at the cost of computational overhead, raising latency and scalability concerns  \cite{lukas2022sok,hu2024inevitable}. 
Watermark removal and model stealing attacks further exacerbate this trade-off. 
Consequently, a watermark that effectively balances   utility,  detectability, robustness, and  efficiency is crucial for  practical LLM deployment.
 
% \phung{revise to here!}

% Designing such a watermark is challenging as a subtle modification to implant a watermark for detectability can degrade text quality \cite{molenda2024waterjudge}. Moreover, increasing watermark strength may enhance robustness but often raises latency and computational overhead, limiting real-time or large-scale deployment. Watermark removal and model stealing attacks further complicate this trade-off. Therefore, a watermark that balances utility, detectability, and robustness to the attacks is essential for LLM deployment.

\textbf{Contributions.}  
This paper proposes \textsc{SafeSeal}, a novel key-conditioned 
LLM watermark  that preserves  utility while ensuring strong and robust detectability with provable guarantees. %To achieve our goals, we answer a fundamental question: \textit{Given a secret key, which words can be substituted to embed a robust watermark   without changing the LLM-generated text's meaning?}

% This paper proposes \textsc{SafeSeal}, a novel key-conditioned post-processing LLM watermark  that minimizes model utility impacts while maintaining high watermark detectability with provable guarantees and robustness against watermark removal and model stealing attacks. 
% To achieve our goals, we answer a fundamental question: \textit{Conditioned on a secret key, which words can be substituted to embed a robust watermark into the LLM-generated text without changing the original meaning?}

% A common insight is that replacing words with context-aware synonyms can preserve semantic meaning. Building upon this, \textsc{SafeSeal} identifies \textit{entities} (e.g.,  time and locations) and \textit{linguistic terms} (e.g., nouns and verbs),  retaining entities while substituting linguistic terms with context-aware synonyms. \textsc{SafeSeal} builds a lookup table of context-aware synonyms using a lightweight LM, then replaces each term with a synonym via key-conditioned Tournament sampling. 
% These steps minimize semantic distortion and hinder adversaries from removing the watermark without damaging text quality, significantly weakening adversarial attacks. 

A key insight is that replacing words with context-aware synonyms can preserve   meaning.  \textsc{SafeSeal} identifies \textit{entities} (e.g., times and names) and  \textit{linguistic terms} (e.g., nouns and verbs), preserving entities while replacing linguistic terms with context-aware synonyms. It builds a synonym lookup table using a lightweight LM and selects substitution via key-conditioned Tournament sampling, minimizing   distortion and making watermark removal difficult without degrading text quality.

For detection, we introduce a key-conditioned contrastive detector that verifies whether text is watermarked under a specific key. It fuses text and key embeddings and is trained on matched and mismatched text-key pairs using classification and contrastive losses. This design enables robust, provider-specific detection of subtle substitutions while resisting removal and model stealing attacks. Keys are tied to users and providers, ensuring only authorized verification and reducing risks of unauthorized detection or reverse-engineering. We further establish a theoretical utility-detectability trade-off, showing that lower synonym similarity improves detectability but reduces utility.

% For detection, we introduce a key-conditioned contrastive detector that verifies text under a specific key by fusing text and key embeddings and training on matched/mismatched pairs. This enables provider-specific detection, supports authorized verification, and reveals a utility--detectability trade-off: lower synonym similarity improves detectability but reduces utility.

% we establish a theoretical connection between   utility and  detectability. Using the lookup table, we bound the expected deviation between original and watermarked outputs and show that lower similarity increases this deviation, improving detectability but reducing utility. %This trade-off offers insights into watermark configurations that achieve strong detection without compromising utility.

 % a smaller similarity threshold increases deviation, enhancing detectability while affecting utility.
 
 % Our analysis reveals that a smaller similarity among the two increases deviation, enhancing detectability while impacting utility.
% Based on that, \textsc{SafeSeal} offers insights to implement watermarks that effectively serve their purpose without compromising utility or user experience.

Experiments on five open-source models \cite{touvron2023llama2openfoundation,jiang2023mistral7b,bi2024deepseek,qwen2.5,team2403gemma} and three benchmarks \cite{dodge2021documentinglargewebtextcorpora,hendrycks2021measuringmassivemultitasklanguage,chen2016thorough} show that \textsc{SafeSeal} preserves $98.3\%$ content similarity and $96.3\%$ entity similarity, with detection rates up to $98.2\%$. It outperforms state-of-the-art watermarks  (KGW \cite{kirchenbauer2023watermark}, EXP \cite{kuditipudi2023robust}, SIR \cite{liu2024semanticinvariantrobustwatermark}, SynthID \cite{dathathri2024scalable}, DTM \cite{munyer2024deeptextmark}, TW \cite{yang2023watermarking}, and LW \cite{he2022protecting}), achieving a $58.47\%$ higher detection rate than the best baseline under removal attacks. Under model stealing attacks, \textsc{SafeSeal} maintains a $72.8\%$ detection rate while forcing noticeable utility degradation, limiting watermark removal without quality loss. Human evaluations further favor \textsc{SafeSeal} by $28\times$ without attacks and $6.5\times$ under attacks. Latency is comparable to the best baseline. We also anonymously release the first Watermark Leaderboard and a real-time demo for transparency and standardized evaluation.

% Extensive experiments on five open-source models \cite{touvron2023llama2openfoundation,jiang2023mistral7b,bi2024deepseek,qwen2.5,team2403gemma} using
%  three benchmark datasets \cite{dodge2021documentinglargewebtextcorpora,chen2016thorough}
%  show that \textsc{SafeSeal}  preserves $98.3\%$ content similarity and $96.3\%$ entity similarity with detection rates up to  $98.20\%$.  It remarkably outperforms the state-of-the-art watermarks (KGW \cite{kirchenbauer2023watermark}, EXP \cite{kuditipudi2023robust}, SIR \cite{liu2024semanticinvariantrobustwatermark}, SynthID \cite{dathathri2024scalable}, DTM \cite{munyer2024deeptextmark}, TW \cite{yang2023watermarking}, and LW \cite{he2022protecting}),  achieving $58.47\%$ higher detection rate than the best baseline in removal attacks. 
%  In model stealing attacks, \textsc{SafeSeal} 
%  maintains a $72.8\%$ detection rate with noticeable utility degradation,  limiting an adversary's ability to remove the watermark without damaging text quality. Human evaluations further confirm its advantages, with \textsc{SafeSeal}  outputs favored $28\times$  more without attacks and $5.57\times$  more under attacks, highlighting strong semantic preservation. Latency remains comparable to other watermarks.  In addition, we (anonymously) release the first Watermark Leaderboard\footnote{\url{https://huggingface.co/spaces/AnonymousResearch/WatermarkLeaderboard}} and a real-time demo watermark interface\footnote{\url{https://huggingface.co/spaces/AnonymousResearch/SafeSeal}} for transparency and standardized evaluation across LLMs, watermarks, and attack settings.  

\textit{Code is available anonymously at:}  
\url{https://github.com/kirudang/SafeSeal}.

\section{Background and Related Work  }
\label{sec:related_work}

% \subsection{Preliminaries and Key Definitions}

\textbf{LLMs and Watermarking.} Given a token vocabulary $\mathcal{V}$ and prompt $x\in\mathcal{V}^T$, an LLM $\theta$ generates $y\in\mathcal{V}^Q$ by modeling $P(y|x)$ ($T,Q > 0$). A watermark embeds imperceptible signals in LLM outputs to verify origin via a generator $\mathcal{W}$ producing $y^{wm}=\theta(x,\mathcal{W})$, and a detector $\mathcal{D}$ that identifies them.
 Watermarked text should preserve \textit{semantic meaning}, i.e., contextual coherence and intent \cite{Devlin2019BERTPO}, and \textit{factual consistency} with the original facts and entities \cite{nan-etal-2021-entity}, since distortions in either reduce utility.

% Let $\mathcal{V}$ denote the token vocabulary. Given a prompt $x=\{x^{(i)}\}_{i=1}^{T}\in\mathcal{V}^{T}$, an LLM $\theta$ generates an output $y=\{y^{(i)}\}_{i=T+1}^{T+Q}\in\mathcal{V}^{Q}$ by modeling the conditional distribution $P(y|x)$ over subsequent tokens. Watermarks embed imperceptible patterns in LLM outputs to verify origin. A watermark consists of (i) a generation function $\mathcal{W}$ producing $y^{wm}=\theta(x,\mathcal{W})$, and (ii) a detector $\mathcal{D}$ that identifies the watermark.

% \textbf{Semantic and Factual Preservation.} Watermarked text should preserve semantic meaning, capturing   contextual coherence and intended meaning \cite{Devlin2019BERTPO} and factual consistency, remaining faithful to original facts and entities \cite{nan-etal-2021-entity}, as distortions in either reduce utility.

% \textbf{Semantic and Factual Preservation.} Watermarked text should preserve \textit{semantic meaning}, i.e., contextual coherence and intended meaning \cite{Devlin2019BERTPO}, and \textit{factual consistency} with the original facts and entities \cite{nan-etal-2021-entity}, since distortions in either reduce utility.

% \textbf{Quality Preservation.} High-quality watermarking must maintain both \textit{semantic meaning} (contextual coherence) and \textit{factual consistency} (faithfulness to entities and facts), as distortions in either degrade utility.

\textbf{Linguistic Sensitivity.} Certain tokens are more sensitive to perturbations. \textit{Named Entity Recognition (NER)} identifies entities (e.g., names, locations, time expressions) that are critical for factual correctness, while \textit{Part-of-Speech (POS)} tagging captures grammatical roles (e.g., nouns, verbs, adjectives, adverbs). Modifying such tokens can disproportionately harm meaning and downstream performance. %Common tools include spaCy~\cite{spaCy2020}, Stanza~\cite{stanza2020}, and NLTK~\cite{bird2009nltk}.

\textbf{Detection Settings.} We consider \textit{cross-provider detection}, where different providers use distinct keys, and \textit{multi-user detection}, where a single provider assigns unique keys per user. In both cases, the detector must correctly attribute the source of watermarked text.

\begin{table*}[t!]
\scriptsize
    \centering
    \caption{Comparison of \textsc{SafeSeal} with representative watermarking techniques.} 
    \renewcommand{\arraystretch}{1.0}
    \begin{tabular}{p{3.12cm} p{1.6cm} p{1.35cm} >{\centering\arraybackslash}m{1.0cm} >{\centering\arraybackslash}m{1.33cm} >{\centering\arraybackslash}m{0.85cm} >{\centering\arraybackslash}m{1.8cm}}
\toprule
\textbf{Method} 
& \textbf{Generation} 
& \textbf{Detection}
& \makecell{\textbf{Provider}\\\textbf{-specific}\\\textbf{Detection}} 
& \makecell{\textbf{Susceptible}\\\textbf{Content}\\\textbf{Preservation}} 
& \makecell{\textbf{Against}\\\textbf{Stealing}\\\textbf{Attacks}} 
& \makecell{\textbf{Utility-Detectability}\\\textbf{Trade-off}\\\textbf{Guarantee}} \\
\midrule
        KGW \cite{kirchenbauer2023watermark}, EXP \cite{kuditipudi2023robust}, SIR \cite{liu2024semanticinvariantrobustwatermark}, SynthID \cite{dathathri2024scalable}             & In-processing     & Statistical test & \ding{51} & \ding{55}                                   & \ding{55} & \ding{55} \\
        WLM \cite{guwatermarking}, Distilled \cite{gu2023learnability}  & In-processing     & Statistical test & \ding{55}   & \ding{55}                                   & \ding{55} & \ding{55} \\
        % SynthID \cite{dathathri2024scalable}                                 & In-processing     & \makecell{Statistical test, Bayesian model}  & \ding{51}  & \ding{55}    & \ding{55} & \ding{55}\\
         GINSEW \cite{zhao2023protecting}, PLMmark \cite{li2023plmmark}                  & In-processing   & Statistical test & \ding{51} & \ding{55}  & \ding{51} & \ding{55} \\
        MorphMark \cite{wang2025morphmark}, Wouters \cite{optimizing2024wouters24a}  & In-processing   & Statistical test  & \ding{51} & \ding{55}                 & \ding{55} & \ding{51} \\
        Token-Specific \cite{hou2024token}                   & In-processing   & Statistical test  & \ding{55}  & \ding{55}                 & \ding{55} & \ding{51} \\
        Hufu \cite{xu2024hufu}, EmMark \cite{zhang2024emmark}    & In-processing   & Model-based   & \ding{51} & \ding{55}        & \ding{51} & \ding{51} \\    
        RLWM \cite{xu2024learning}    & In-processing   & Model-based   & \ding{55} & \ding{55}        & \ding{55} & \ding{55} \\ 
        Cross-Attention WM \cite{baldassini2024cross}    & In-processing   & Model-based   & \ding{51} & \ding{55}        & \ding{55} & \ding{55} \\  
        DTM \cite{munyer2024deeptextmark}               & Post-processing   & Model-based   & \ding{55} & \ding{51}               & \ding{55} & \ding{55}\\
        TW \cite{yang2023watermarking}, Post-hoc WM \cite{hao2025post}                  & Post-processing   & Statistical test & \ding{55} & \ding{51}                & \ding{55} & \ding{55} \\
        LW \cite{he2022protecting}                & Post-processing   & Statistical test & \ding{55} &  \ding{51}                 & \ding{51} & \ding{55} \\
        CATER \cite{he2022cater}                 & Post-processing   & Statistical test & \ding{51}  &  \ding{51}                 & \ding{51} & \ding{55} \\
        PostMark  \cite{chang2024postmark}                 & Post-processing   & Statistical test & \ding{55} &  \ding{51}                 & \ding{55}  & \ding{55} \\
        \midrule
        \textbf{\textsc{SafeSeal} (Ours) } & Post-processing & Model-based   & \ding{51} & \ding{51} & \ding{51} & \ding{51} \\ 
        \bottomrule
    \end{tabular}
    \label{tab:Summary}
\end{table*} 
\setlength{\textfloatsep}{10pt}

\textbf{LLM Watermarking Techniques.}
Existing LLM watermarks can fall  into \textit{1) in-processing} methods, which embed signals during decoding by modifying logits, token distributions, or training~\cite{kirchenbauer2023watermark,zhao2024provable,liu2024adaptive}, and \textit{2) post-processing} methods, which alter generated text through synonym substitution, rule-based replacement, or paraphrasing~\cite{dathathri2024scalable,giboulot2024watermax,hou2023semstamp}. %Detection typically relies on statistical tests or classifiers. 
Although effective, both categories may degrade utility by altering semantic meaning or factual consistency, especially when entities are modified.

Prior work often defines ``distortion-free'' watermarking as preserving the next-token distribution~\cite{dathathri2024scalable,kuditipudi2023robust,hou2024k,christ2024undetectable}, but distributional preservation does not guarantee semantic or factual consistency. For example, EXP~\cite{kuditipudi2023robust} and SynthID~\cite{dathathri2024scalable} can still introduce entity shifts or fabricated details due to cascading effects in autoregressive generation (Table~\ref{tab:non_distortion}, Appx.\ref{sec:related_work}). Similarly, utility-detectability trade-offs are often studied through proxy metrics such as sampling strength, entropy, or Pareto objectives~\cite{kirchenbauer2023watermark,huang2024waterpool,wang2025trade}, which may miss such distortions. Moreover, many methods rely on secret keys or provider-specific triggers \cite{hu2024unbiased,xu2024hufu,zhang2024emmark}, but assume simplified single-key settings, whereas real deployments require robust cross-provider and multi-user detection under unique private keys.

Table \ref{tab:Summary} compares \textsc{SafeSeal} with related work. \textsc{SafeSeal} is a key-conditioned post-processing watermark that preserves entities and replaces linguistic terms with context-aware synonyms to maintain \textit{utility}, defined as  semantic meaning, factual consistency, and downstream utility. Its contrastive detector captures key-specific patterns for robust  cross-provider and multi-user detection, while our evaluation covers utility, robustness,  and the theoretical utility-detectability trade-off.

\section{Threat Models}

\textbf{Setting.} This work considers a service provider offering a proprietary LLM $\theta$ via an API. 
Users   submit a prompt $x$ 
 and receive an   output $y$ (or $y^{wm}$ if watermarked).  In this work, IP protection refers to preventing unauthorized use or  stealing attacks via exploiting the outputs generated by $\theta$.

\textbf{Adversary's Capabilities.} 
% The adversary can query multiple prompts via the API and receive corresponding outputs, but lacks knowledge of the LLM's specifications. Also, the LLM does not disclose any details in its responses, preventing inadvertent leaks.
We employ a practical black-box setting where the adversary can query the API but has no access to the proprietary LLM's internals, preventing leakage via the outputs $y$. 
 %We employ a practical black-box setting, where the adversary can query the API  but has no knowledge of the proprietary LLM's internal structure and parameters,  preventing inadvertent information leaks via the output $y$. 
 %ensuring no inadvertent leaks through the output $y$. 

% \textbf{Adversary's Goals.}  We consider two common attacks, including \textbf{(1) Watermark removal}, where the adversary uses  popular removal techniques, such as paraphrasing, word insertation or synonym substitution \cite{krishna2024paraphrasing,pan2024markllm,zhang2023watermarks,lin2021towards} to produce a watermark-removed output $\bar{y}$ that bypasses detection; and \textbf{(2) Model stealing}, where  the adversary trains its surrogate model $\theta_{adv}$ to replicate  $\theta$ \cite{sander2024watermarking,carlini2024stealing,wu2024bypassing,jovanovic2024watermark,pasquini2025llmmap,zeng2024huref,wallace2020imitation,xu2022student,krishna2019thieves,birch2023model}. The surrogate model can be a simplified or similar version of the service provider's LLM. These attacks pose significant and practical  threats to service providers, as adversaries can bypass watermarks and steal model behaviors, compromising IP protection and traceability and potentially enabling competing services and revenue loss.

\textbf{Adversary's Goals.} We consider two common attacks: \textit{watermark removal}, where an adversary uses paraphrasing, insertion, or synonym substitution \cite{zhang2023watermarks,lin2021towards} to produce an watermark-removed  output $\bar{y}$  that bypasses detection; and \textit{model stealing}, where the adversary trains a surrogate model $\theta_{adv}$ to imitate $\theta$ \cite{carlini2024stealing,wu2024bypassing,jovanovic2024watermark,pasquini2025llmmap,wallace2020imitation}. These attacks threaten IP protection, traceability, and service integrity.

\section{\textsc{SafeSeal} LLM Watermarking}
\label{sec:watermark_mechanism}
 
This work aims to develop an effective and practical watermark that minimizes utility impacts via minimal semantic distortion, achieves high detectability, resists removal and model stealing attacks, incurs low latency, and provides provable utility-detectability guarantees. To this end, we ask: \textit{Which tokens can be substituted to embed a robust watermark without changing the original meaning?}

 % This work aims to develop an effective and practical watermark that  1) minimizes utility impacts via minimal semantic distortion, 2) achieves high detectability, 3) resists watermark removal and model stealing attacks, 4) incurs low latency, and 5) provides provable guarantees on the utility-detectability trade-off.  
 %  To achieve these goals, we answer a fundamental question:  \textit{Which tokens can be substituted to embed a robust watermark into the LLM-generated text  without changing the original meaning?} 

  %\textit{How do we determine which words to modify and implant a watermark in text generated by an LLM without distorting its original meaning?}
  
  % : \textit{1) Given an LLM output, which words (i.e., tokens) should be modified to inject a watermark without altering its meaning?} \textit{2) How can we effectively detect our watermark?} and \textit{3) How can we theoretically establish and efficiently control the trade-off between utility and detectability?}

Our solution includes: 1) Preserving entities  while replacing linguistic terms with key-conditioned, context-aware synonyms via a lookup table and  Tournament sampling
(\textit{Section \ref{subsec:WMgen}}). These steps preserve semantic meaning and factual consistency while enabling provider-specific detection; 2) Developing a key-conditioned, contrastive detector to improve robustness (\textit{Section \ref{subsec:WMdetect}}); 3) Bounding output deviation via token-level substitutions to control the utility-detectability trade-off (\textit{Section \ref{subsec:theory_analysis}});  and 4) Reducing latency via  parallelization, batching, caching, and employing lightweight models (\textit{Section \ref{subsec:latency}}). Alg.~\ref{alg:watermark_ldp} (Appx.~\ref{app:additional_results})  and Fig.~\ref{fig:Wm_process} summarize  the pseudo-code and design of \textsc{SafeSeal}.

 \textbf{Example.}  Table~\ref{table:sample} compares LLaMA-2 outputs   watermarked by KGW, SynthID, and \textsc{SafeSeal}. KGW and SynthID exhibit notable constraints as they fail to preserve entities, e.g., KGW retains only 2 of 7 entities  (yellow text)  and causes semantic distortions  or fabricated details (red text), e.g., ``Mr. Rose,''  ``Wal-Mart,'' and ``28-hour.''  In contrast, \textsc{SafeSeal} \textbf{preserves entities}, i.e., retaining all $7$ entities,  and \textbf{substitutes linguistic terms  with context-aware synonyms} (green text), such as replacing ``decided'' with \{``agreed,'' ``said,'' ``addressed''\}.  It further improves robustness through key-conditioned Tournament sampling while maintaining semantic meaning and factual consistency.

 \begin{figure*}[t]
      \centering    
      \includegraphics[scale=0.047]{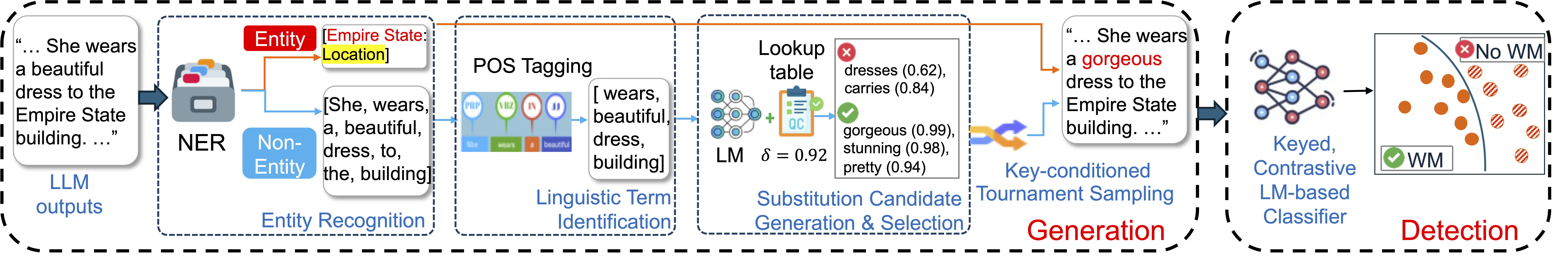} %\vspace{-10pt}
      \caption{\textsc{SafeSeal} overview.} %\vspace{-10pt}
      \label{fig:Wm_process}
 \end{figure*}

\begin{table*}[t!]
\footnotesize
\centering
\caption{\textsc{SafeSeal} and other watermarks on LLaMA-2. Yellow marks \hlyellow{entities}, red shows \hlred{semantic changes}, green is \hl{preserved meaning}, and dark green is \hlg{word changes with meaning preserved}.}
% \caption{Comparison of \textsc{SafeSeal} and other watermarks on LLaMA-2. Yellow highlights \hlyellow{entities}, red indicates \hlred{semantic changes}, green denotes \hl{preserved meaning}, and dark green marks \hlg{word changes that preserve meaning}.}
% (Yellow text highlights the key information to protect, green text indicates preserved information, and red text shows modified key information)}
\begin{tabular}{|p{0.09\linewidth}|p{0.85\linewidth}|}
\hline
\centering  \textbf{\textbf{Output}} & \centering \textbf{Content} \tabularnewline \hline
% \centering 
% \textbf{Prompt} & 
% \raggedright ...Ms. Rinaldi ... as a Senior Librarian at the Ocean County Library, Long Beach Island Branch in Surf City, New Jersey. ... 
% \tabularnewline \hline
\centering  \multirow{2
}{*}{\makecell{\textbf{No} \\ \textbf{WM}}}  & 
\raggedright   \hlyellow{QVC} \hl{decided} to get  \hl{action}, offering \hl{deals} on \hl{high-definition TVs, home appliances and more}. ``We’re not missing the   \hl{fun}'', \hlyellow{Rose} \hl{said}.  \hlyellow{QVC}   \hl{start} its \hlyellow{Black Friday} \hl{sales} at \hlyellow{12:01 a.m}. \hlyellow{Thanksgiving Day}, and \hlyellow{James}  with a \hl{marathon of deals,  breaks and giveaways}. [...] \tabularnewline \hline
\centering \multirow{4
}{*}{\makecell{\textbf{KGW} \\ \cite{kirchenbauer2023watermark} }} & 
\raggedright \hlred{``We're not a mall. We don't have to compete with the crowds''} \hlred{Mr.} \hlyellow{Rose} \hl{said} \hlred{in a telephone interview.} ``\hlred{We're a place where people can shop} at their \hl{leisure}''. \hlyellow{Black Friday} is usually associated with \hlred{big-box stores like Wal-Mart and Target}, which open \hl{early in the morning} and offer \hl{steep discounts} [...] %on a \hlred{limited selection of items}.   
\tabularnewline \hline
\centering
\multirow{3}{*}{\makecell{\textbf{SynthID} \\ \cite{dathathri2024scalable}  }} & 
\raggedright  \hlyellow{QVC}  is \hlred{taking a different tack}, \hlred{Rose  said}, because it wants to \hlred{get a jump on the holiday season} with \hlred{a boost} from the \hlyellow{Black Friday}  \hlred{bargain mania}. \hl{``We want to be the alternative to malls''}, \hlyellow{Rose} \hl{said}. \hlyellow{QVC}'s \hlred{28-hour}  \hl{marathon of deals and new items} kicks off on \hlyellow{Thanksgiving} \hlred{night}   and will feature \hl{discounts} on \hl{high-profile electronics}, and \hlred{free shipping}. [...] 
\tabularnewline \hline
\centering \multirow{2}{*}{\makecell{\textbf{\textsc{SafeSeal}} \\ \textbf{(Ours)}}} & 
\raggedright \hlyellow{QVC}  \hlg{agreed} to get   \hlg{initiative}, offering \hlg{discounts} on \hl{high-definition TVs, home accessories and more}. \hl{``We’re not missing the fun''}, \hlyellow{Rose} \hlg{agreed}. \hlyellow{QVC}   \hl{start} its \hlyellow{Black Friday} \hlg{deals} at \hlyellow{12:01 a.m}. \hlyellow{Thanksgiving Day}, and \hlyellow{James}  with \hl{a marathon of  }\hlg{sales}\hl{, }\hlg{changes}\hl{ and  }\hlg{promotions}. [...] 
\tabularnewline \hline
\end{tabular}
\label{table:sample}
\end{table*}

\subsection{Key-Conditioned Watermark Generation}  
\label{subsec:WMgen}

\textbf{Watermarkable Token Identification.}
Given an output $y$, we first use NER \cite{spaCy2020} to identify entities $e\in\mathcal{E}$ and exclude them from watermarking to preserve semantic and factual consistency. We then apply POS tagging \cite{petrov2011universal,loper2002nltk} to identify \textit{linguistic terms} critical to text meaning and structure, selecting \textit{nouns, verbs, adjectives,} and \textit{adverbs}  as the watermarkable token set $T^{wm}$. These POS types provide more flexibility and broader synonym  substitution   than entities or other linguistic categories. Table~\ref{tab:pos_categories} (Appx.~\ref{appendix:Linguistic}) summarizes the POS types used in \textsc{SafeSeal}.

% \textbf{Entity Recognition.}   
% Given an output $y$, we use an NER framework \cite{spaCy2020}   to recognize  entities $e \in \mathcal{E}$ in $y$. We exclude entities from watermarking to avoid semantic distortion and factual inconsistency. 

% \textbf{Identification of Linguistic Terms.}   
% We identify \textit{linguistic terms} critical for text understanding and structure using POS tagging \cite{petrov2011universal,loper2002nltk}.  %, we classify words into different linguistic types. 
%  Unlike prior watermarks that target only adjectives \cite{li2023protecting} or exclude numerics and stopwords \cite{munyer2024deeptextmark}, we use \textit{nouns, verbs, adjectives,} and \textit{adverbs} to form a watermarkable token set $T^{wm}$. These types provide more flexibility and broader synonyms than entities or other linguistic types. Table \ref{tab:pos_categories}  (Appx.~\ref{appendix:Linguistic}) summarizes the POS types used in  \textsc{SafeSeal}.

\textbf{Lookup Table Generation.} Next,  we build a lookup table $\mathcal{L}$ of context-aware synonyms. For each $t_i\in T^{wm}$, $\mathcal{L}_{t_i}$ stores substitution candidates derived from the overlapping context window $C=\{t^{1},\dots,t_i,\dots,t^{N}\}$ containing $t_i$\footnote{Superscripts and subscripts distinguish normal and watermarkable tokens in $C$.}. In our experiments, $C$ spans about $200$ tokens. Candidates are generated and selected to preserve the meaning and grammatical form of $t_i$ in context, as follows.

% For each   $t_i \in T^{wm}$,  entry $\mathcal{L}_{t_i}$ contains  substitution candidates. Without loss of generality, we describe the construction of $\mathcal{L}_{t_i}$; the same procedure applies to all tokens in $T^{wm}$ to form $\mathcal{L}$.
%  To derive  candidates for  $t_i$, we partition $y$ into overlapping context windows  and use the window containing $t_i$, denoted as  $C=\{t^{1},\dots,t_i,\dots,t^{N}\}$\footnote{ 
% Superscripts and subscripts distinguish normal and watermarkable tokens in $C$.}.  In our experiments, $C$ spans approximately $200$ tokens to provide sufficient surrounding context. Candidates must preserve the meaning and grammatical form of $t_i$ within $C$. They are constructed through generation and selection, as follows. 

\textit{{Substitution Candidate Generation.}} 
We use a lightweight LM $\mathcal{A}$ to generate context-aware synonyms for each $t_i$. We mask $t_i$ in its context window $C$ to obtain $C_M$, then feed $C \oplus [SEP] \oplus C_M$ with an instruction to produce a ranked candidate list $\mathcal{L}_{t_i}$. We further filter candidates by removing duplicates, antonyms, incomplete terms, and non-WordNet entries~\cite{miller1995wordnet}, while enforcing POS consistency with NLTK~\cite{loper2002nltk} and contextual fit (e.g., replacing "an apple" with "a fruit" rather than "an fruit").

%  We use a lightweight LM $\mathcal{A}$ to generate context-aware, syntactically similar synonyms for each $t_i$. Specifically, we mask $t_i$  with  $\langle MASK \rangle$ in its context window $C$ to obtain $C_M$, concatenate $C \oplus [SEP] \oplus C_M$, and feed it with an instruction into $\mathcal{A}$ to produce a ranked candidate list $\mathcal{L}_{t_i}$. For example, for $C=$ ``She wears a beautiful dress.'' and $t_i=$ ``wears,'' the input is ``She wears a beautiful dress. $[SEP]$ She $\langle MASK \rangle$ a beautiful dress.'' The instruction to $\mathcal{A}$ is ``Given the context, replace $\langle MASK \rangle$ with a grammatically correct alternative that preserves meaning and ensures natural flow:''.  
% To further improve the quality of $\mathcal{L}_{t_{i}}$, we apply  post-processing steps, including removing duplicates, antonyms,  incomplete terms, and candidates not in   WordNet   \cite{miller1995wordnet}; and ensuring grammatical consistency with $t_i$ through POS alignment using NLTK \cite{loper2002nltk}.  We also enforce contextual and syntactic compatibility, ensuring that substitutions fit naturally within the sentence (e.g., replacing "an apple" with "a fruit" rather than "an fruit").

\textit{Substitution Candidate Selection.}
To preserve semantics, we keep candidates whose sentence-level similarity $\mathbb{S}_{ij}$ between the original sentence $S_{t_i}$  and substituted sentence $S_{t_j}$ (where $t_i$ is substituted by $t_j$) exceeds threshold $\delta$. Tokens with fewer than two valid candidates are skipped. Formally, for $t_i \in T^{wm}$: $\mathcal{L}_{t_i}^{\delta} =  \big\{ t_j \colon \mathbb{S}_{ij} \ge \delta \big\} \text{ and } |\mathcal{L}_{t_i}^{\delta}| \geq 2$.
The final table is $\mathcal{L}=\{\mathcal{L}_{t_i}^{\delta}\}_{t_i\in T^{wm}}$. % balances semantic fidelity and substitution diversity.
 
% \textit{{Substitution Candidate Selection.}} 
% To preserve semantics, we retain from the candidate set only those substitutions whose  sentence-level similarity score $\mathbb{S}_{ij}$ 
% between  the original sentence $S_{t_i}$ (containing $t_i$) and  the modified sentence $S_{t_j}$ (where $t_i$ is substituted by $t_j$) meets a threshold $\delta$. If fewer than two candidates remain, the original token $t_i$ is skipped. Formally, for $t_i \in T^{wm}$:
% \begin{equation}
% \small
%    \mathcal{L}_{t_i}^{\delta} =  \big\{ t_j \colon \mathbb{S}_{ij} \ge \delta \big\} \text{ and } |\mathcal{L}_{t_i}^{\delta}| \geq 2,
% \end{equation}
% where final lookup table $\mathcal{L} = \big\{\mathcal{L}_{t_i}^{\delta} \big\}_{t_i \in T^{wm} }$  balances semantic fidelity with substitution diversity.

 \textbf{Key-Conditioned Watermark Injection.} 
 To select a substitute $t_j$ from $\mathcal{L}_{t_i}^{\delta}$, we use key-conditioned Tournament sampling \cite{dathathri2024scalable}. Candidate similarities $\mathbb{S}_{ij}$ are converted into sampling probabilities: $P_i(t_j)=
\frac{\exp(\alpha\mathbb{S}_{ij})}
{\sum_{j=1}^{|\mathcal{L}_{t_i}^{\delta}|}\exp(\alpha\mathbb{S}_{ij})}$, 
% \begin{equation}
% \small
% P_i(t_j)=
% \frac{\exp(\alpha\mathbb{S}_{ij})}
% {\sum_{j=1}^{|\mathcal{L}_{t_i}^{\delta}|}\exp(\alpha\mathbb{S}_{ij})},
% \end{equation}
where  $\alpha> 0$ controls the utility-detectability trade-off, in which larger values favor high-similarity substitutes, while smaller values increase diversity.  
Given $P_i$, we sample $M=c^m$ candidates with replacement and run an $m$-round tournament of match size $c\ge2$. In each round $r$, candidates receive keyed pseudorandom scores $g_r(t_j)=\mathrm{PRF}(k,h,t_j,r)\in\{0,1\}$, where $k$ is the secret key and $h$ is the context. The final winner is selected as the substitute for $t_i$.

\subsection{Key-Conditioned Watermark Detection} 
\label{subsec:WMdetect} 

Unlike rule-based or statistical detectors~\cite{kirchenbauer2023watermark,liu2024semanticinvariantrobustwatermark,he2022protecting} that are vulnerable to paraphrasing and edits, or model-based detectors \cite{xu2024learning,hou2024token,dasgupta2025watermarking,munyer2024deeptextmark} that improve robustness but  lack provider-specific detection, we formulate detection as key-conditioned binary classification. Given text $y$ and key $k$, the detector $\mathcal{D}$ predicts whether $y$ is watermarked under $k$. Positive pairs contain text watermarked with $k$, while negative pairs are non-watermarked text or text watermarked with another key $k'\neq k$.

We extract detector features by fusing text, key, and pairwise interaction embeddings. Text $y$ is encoded with a pretrained bidirectional transformer $f_{enc}$ (e.g., RoBERTa \cite{liu2019roberta}) and mean pooling to obtain $\mathbf{h}_y\in\mathbb{R}^H$. For key $k$, we avoid trainable embeddings to reduce memorization and leakage \cite{song2020information,wan2024information} by using a PRF that maps $k$ to $\mathbf{v}\in[0,1]^H$, which is rescaled to $\tilde{\mathbf{v}}=2(\mathbf{v}-0.5)\in[-1,1]^H$ and passed through a frozen MLP $f_{MLP}$ to obtain $\mathbf{h}_k=f_{MLP}(\tilde{\mathbf{v}})$. We obtain the element-wise Hadamard product $\mathbf{h}_y \odot \mathbf{h}_k$  and the absolute difference $|\mathbf{h}_y-\mathbf{h}_k|$. We   form the fused representation:
\begin{equation}
\small
\boldsymbol{\phi}=
[\mathbf{h}_y;\mathbf{h}_k;\mathbf{h}_y\odot\mathbf{h}_k;|\mathbf{h}_y-\mathbf{h}_k|]\in\mathbb{R}^{4H}.
\end{equation}
These features capture text-key alignment and mismatch. The fused vector $\boldsymbol{\phi}$ is passed to a two-layer MLP $f_{out}$ (from the same backbone as $f_{enc}$ and $f_{MLP}$)  to produce $z=f_{out}(\boldsymbol{\phi})$. Finally, the probability that $y$ is watermarked under key $k$: $P(\mathcal{D}(y|k)=1)=\sigma(z)$.

To train the detector, we  optimize detection and contrastive losses. With  $B$ is the batch size, $\mathbb{I}_b$ indicates whether $y_b$ is watermarked under key $k_b$,  the detection loss is a binary cross-entropy, as:
\begin{equation}
\small
\mathcal{L}_{\mathrm{det}}=
-\frac{1}{B}\sum_{b=1}^{B}
\Big[
\mathbb{I}_b\log\sigma(z_b)+(1-\mathbb{I}_b)\log(1-\sigma(z_b))
\Big].
\end{equation}
% where $B$ is the batch size, $\mathbb{I}_b$ indicates whether $y_b$ is watermarked under key $k_b$.

To learn key-specific alignment, we use symmetric InfoNCE \cite{jia2021scaling,radford2021learning,uesaka2024weighted}
over matched text-key pairs $\{(y_i,k_i)\}_{i=1}^B$:
\begin{equation}
\small
\mathcal{L}_{\mathrm{con}}
=
\tfrac{1}{2}
\left(
\mathcal{L}_{t\to k}
+
\mathcal{L}_{k\to t}
\right),
\end{equation}
where $\mathcal{L}_{t\to k}
    =
    -\frac{1}{B}\sum_{i=1}^{B}
    \log
    \frac{
        \exp\big(\mathbb{S}_{\mathrm{cos}}(y_i,k_i)\big)
    }{
        \sum_{j=1}^{B}
        \exp\big(\mathbb{S}_{\mathrm{cos}}(y_i,k_j)\big)
    }$  aligns each text $y_i$   with its matching key $k_i$ and $\mathcal{L}_{k\to t}
    =
    -\frac{1}{B}\sum_{j=1}^{B}
    \log
    \frac{
        \exp\big(\mathbb{S}_{\mathrm{cos}}(y_j,k_j)\big)
    }{
        \sum_{i=1}^{B}
        \exp\big(\mathbb{S}_{\mathrm{cos}}(y_i,k_j)\big)
    }$ 
    align each key $k_j$    with its matching text $y_j$.   %$(y_i, k_i)$ is a positive pair and $(y_i, k_j)$ or $(y_j, k_i)$ are negative pairs ($j\neq k$). 
    $\mathbb{S}_\mathrm{cos}(y,k)$ denotes   cosine similarity between text $y$ and key $k$. Here, $\mathcal{L}_{\mathrm{det}}$ classifies watermark presence under key $k$, while $\mathcal{L}_{\mathrm{con}}$ aligns matched text-key pairs and separates mismatched ones, enabling provider-specific detection. 
    % each direction pulls matched pairs together and pushes mismatched pairs apart using cosine similarity. 
    The final objective is:
\begin{equation}
\small
\mathcal{L}
=
\mathcal{L}_{\mathrm{det}}
+
\lambda\mathcal{L}_{\mathrm{con}},
\quad \lambda\in[0,1].
\end{equation}

\subsection{Utility and Detectability Bounds}  
\label{subsec:theory_analysis}
% We rely on a key observation that a larger deviation between the original and watermarked outputs reduces model utility but improves watermark detectability, since the two become easier to  distinguish. 
% Conversely, smaller deviation preserves utility but make detection harder. 
% Similarity scores, such as cosine, BERTScore \cite{zhang2020bertscore},  and METEOR \cite{banerjee2005meteor}, typically lie in  $[0,1]$, so  we assume that the similarity score $\mathbb{S}(\cdot,\cdot) \in [0,1]$ and define a distance  as $ \mathbb{D} = 1-\mathbb{S}(\cdot,\cdot)$.  Measuring differences between two text  is challenging due to linguistic complexity and context-dependence. To simplify, we approximate this distance by averaging the difference between   original and substituted tokens within their context, as  follows: 

% We observe that larger deviations between original and watermarked outputs improve detectability but reduce utility, while smaller deviations preserve utility but hinder detection. 

We observe that utility and detectability depend on the deviation between original and watermarked text, as larger deviations improve detectability but reduce utility. We define this deviation as $\mathbb{D}=1-\mathbb{S}(\cdot,\cdot)$, where $\mathbb{S}\in[0,1]$ is a similarity score. Since full-text distance is difficult to measure, we approximate $\mathbb{D}$ by averaging contextual differences between original and substituted tokens:
% We observe that utility and detectability are highly correlated with the deviations between between watermarked and original text, as larger deviations improve detectability but reduce utility and vice versa. To derive the correlation between utility and detectability, we first define distance $\mathbb{D} = 1-\mathbb{S}(\cdot,\cdot)$, where similarity score $\mathbb{S}(\cdot,\cdot)$, which typically lies in $[0,1]$.  
% Measuring text differences is challenging due to linguistic complexity, so we approximate this distance by averaging differences between original and substituted tokens in context:
{\small
\begin{align}
    \Delta(y, y^{wm})  = \frac{1}{ |T^{wm}|} \sum_{t_i \in T^{wm}, t_j  \in \mathcal{L}_{t_i}^{\delta}}  \mathbb{D}(t_i, t_j), \label{eq:sentence_deviation}
\end{align}
}

\noindent where $\Delta(y, y^{wm})$ is the deviation between outputs $y$ and $y^{wm}$; $T^{wm}$ is the watermarkable token set; and $\mathcal{L}_{t_i}^{\delta}$ is  the  lookup entry of $t_i \in T^{wm}$. After watermarking, each token $t_i \in T^{wm}$  is  replaced by a  context-aware synonym   $t_j$ in   $\mathcal{L}_{t_i}^{\delta}$.   Intuitively,  a larger value of  $\Delta(y, y^{wm})$   implies greater distortion of $y^{wm}$ relative to $y$, degrading utility but improving detectability. 
% To study this utility-detectability trade-off, we determine the upper and lower bounds for LLM's utility and    detectability  by finding the range for the deviation $\Delta(y, y^{wm})$. We approximate the range using the expected deviation and derive its upper and lower bounds, as follows:
To analyze this trade-off, we bound utility and detectability by characterizing the range of $\Delta(y, y^{wm})$. We approximate this range via the expected deviation and derive its upper and lower bounds as follows:

\begin{theorem}
    For a  similarity threshold $\delta$, the expected deviation between $y$ and $y^{wm}$ is bounded:
    \begin{align}
    \frac{d_{lb}}{|T^{wm}|}    \leq \mathbb{E}\Big(\Delta(y, y^{wm})\Big) \leq  \frac{d_{ub}}{|T^{wm}|},
    \label{eq:boundsK}
\end{align}
% \noindent where $|T^{wm}|$ is the size of the watermarkable token set in $y$ for watermarking, directly affected by $\delta$, and $d_{lb}$ and $d_{ub}$ represent the minimum and maximum values for  the distance $\mathbb{D}(t_i, t_j)$ between two distinct tokens $t_i$ and $t_j$.
\noindent where $|T^{wm}|$ is the number of watermarkable tokens in $y$ (controlled by $\delta$), and $d_{lb}$ and $d_{ub}$ denote the minimum and maximum distance $\mathbb{D}(t_i, t_j)$ between distinct tokens $t_i$ and $t_j$.
\label{theorem:bounds}
\end{theorem}

Detailed proof of Theorem \ref{theorem:bounds} is provided in Appx.~\ref{app:theorem}. Theorem \ref{theorem:bounds} characterizes the utility–detectability trade-off via upper and lower bounds on the expected deviation, as follows:

% \textbf{(i)} The deviation lower bound $\frac{d_{lb}}{|T^{wm}|}    \leq \mathbb{E}\Big(\Delta(y, y^{wm})\Big)$ represents the worst-case scenario for detectability and the best-case scenario for  utility. Minimizing deviation reduces  differences between  original and watermarked outputs, making it more difficult for the  detector to distinguish between the two outputs, thereby reducing detectability. The equality occurs when all synonyms are semantically similar to corresponding watermarkable tokens, resulting in $d_{lb}=0$ and zero deviation. Under our sampling mechanism, this rare case occurs with probability: $1/ \Big(\frac{|\mathcal{L}|}{|T^{wm}|}\Big)^{|T^{wm}|}$, 
% where $\frac{|\mathcal{L}|}{|T^{wm}|}$ is the average size of all entries  $\mathcal{L}_{t_i}$ for the token $t_i$ in the lookup table $\mathcal{L}$ in $y$.
% For instance, with $\delta=0.92$, we have $|T^{wm}|=23$ and average lookup entry size $\frac{|\mathcal{L}|}{|T^{wm}|} = 6$; the probability is $\frac{1}{6^{23}} \approx 1.23 \times 10^{-18}$, which is extremely small. This probability increases as $|T^{wm}|$ decreases (or $\delta$ increases).

  \begin{figure}[t]
\centering
\subfigure[Deviation across $\delta$ values]{\includegraphics[scale=0.24]{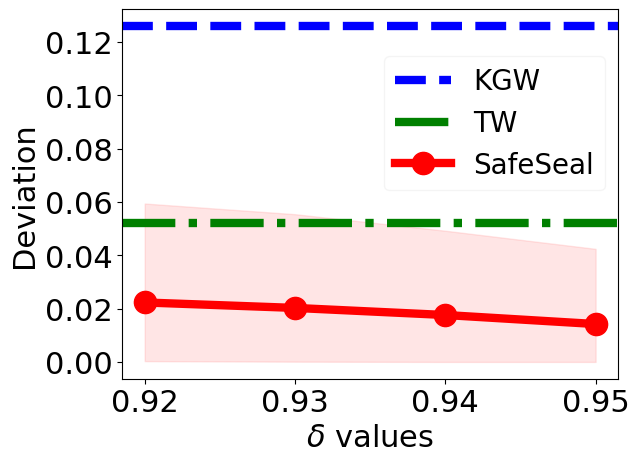}\label{fig:Theory-main-a}}\hspace{2cm}
\subfigure[Deviation, text length across $|T^{wm}|$]{\includegraphics[scale=0.24]{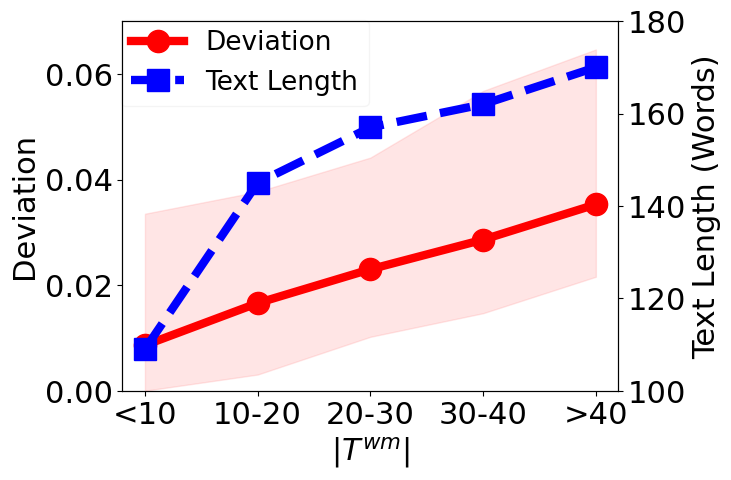}\label{fig:Theory-main-b}}
\caption{Impact of similarity threshold $\delta$, watermarkable set size $|T^{wm}|$, and  correlation with text length on LLaMA-2. (Solid lines show  expected deviations; shaded areas are two bounds)} \vspace{-7.5pt}
\label{fig:Theory-main}
\end{figure} 

\textbf{(i)} The lower bound $\frac{d_{lb}}{|T^{wm}|}\le \mathbb{E}[\Delta(y,y^{wm})]$ corresponds to best-case utility but weakest detectability, with equality when substitutions are perfectly similar ($d_{lb}=0$). Under our sampling, this rare event occurs with probability $1/\big(\frac{|\mathcal{L}|}{|T^{wm}|}\big)^{|T^{wm}|}$; e.g., for $\delta=0.92$, $|T^{wm}|=23$, and average entry size $6$, this is $1/6^{23}\approx1.23\times10^{-18}$. The upper bound $\mathbb{E}[\Delta(y,y^{wm})]\le\frac{d_{ub}}{|T^{wm}|}$ captures worst-case utility, where least-similar substitutions maximize deviation, an equally rare extreme.
 %an equally rare extreme of probability $1/ \Big(\frac{|\mathcal{L}|}{|T^{wm}|}\Big)^{|T^{wm}|}$. 

% \textbf{(i)} The lower bound $\frac{d_{lb}}{|T^{wm}|} \le \mathbb{E}[\Delta(y, y^{wm})]$ reflects best-case utility but worst-case detectability: smaller deviation makes outputs harder to distinguish. Equality holds when all substitutions are perfectly similar ($d_{lb}=0$), yielding zero deviation. Under our sampling mechanism, this rare event occurs with probability $1/\big(\frac{|\mathcal{L}|}{|T^{wm}|}\big)^{|T^{wm}|}$, where $\frac{|\mathcal{L}|}{|T^{wm}|}$ is the average lookup entry size. For example, when $\delta=0.92$, $|T^{wm}|=23$ and the average entry size is $6$, yielding $\frac{1}{6^{23}} \approx 1.23 \times 10^{-18}$. 
% By contrast, the upper bound $\mathbb{E}\Big(\Delta(y, y^{wm})\Big) \leq  \frac{d_{ub}}{|T^{wm}|} $ represents worst-case utility, where substitutions are least similar, maximizing the deviation. This extreme is also rare, with probability $1/ \Big(\frac{|\mathcal{L}|}{|T^{wm}|}\Big)^{|T^{wm}|}$.  

\textbf{(ii)}  When both detectability and utility outperform baselines, \textsc{SafeSeal} achieves a better trade-off. Fig.~\ref{fig:Theory-main-a} shows lower expected deviations across $\delta$ on LLaMA-2, e.g., KGW and TW reach deviations of $0.127$ and $0.052$, respectively, $2.26$--$5.52\times$ higher than \textsc{SafeSeal} ($0.023$ at $\delta=0.92$), with gaps up to $8.93\times$ at larger $\delta$. Similar trends hold on Mistral, up to $13.22\times$ (Fig.~\ref{fig:T_Length_Mistral}, Appx.~\ref{app:additional_results}).

% \textbf{(ii)} Under these worst-case bounds, when both detectability and utility outperform baselines, \textsc{SafeSeal} achieves a better trade-off. 
% Fig.~\ref{fig:Theory-main-a} shows that \textsc{SafeSeal} consistently yields lower expected deviations across  $\delta$ on LLaMA-2. For example, KGW and TW produce deviations of $0.127$ and $0.052$, $2.26$–$5.52\times$ higher than \textsc{SafeSeal} ($0.023$ at $\delta=0.92$), with the gap reaching $8.93\times$ at higher $\delta$. Similar trends hold on Mistral (up to $13.22\times$), detailed in Fig.~\ref{fig:T_Length_Mistral} (Appx.~\ref{app:additional_results}), showing \textsc{SafeSeal} achieves a strong utility–detectability trade-off under bounded worst cases.

 \textbf{(iii)} The bounds depend on average lookup size $\frac{|\mathcal{L}|}{|T^{wm}|}$, watermarkable token count $|T^{wm}|$, threshold $\delta$, and $d_{lb},d_{ub}$.  For a given $y$, $|T^{wm}|$ is set by $\delta$. Higher $\delta$ restricts candidates, lowering $d_{ub}$ and improving utility but reducing detectability; lower $\delta$ expands candidates, increasing deviation and detectability at the cost of semantic fidelity.

% \textbf{(iii)} The bounds depend on the average lookup size $\frac{|\mathcal{L}|}{|T^{wm}|}$, watermarkable token count $|T^{wm}|$, similarity threshold $\delta$, and bounds $d_{lb}, d_{ub}$. For a given $y$, $|T^{wm}|$ is set by $\delta$. Both $d_{lb}$ and $d_{ub}$ vary with these factors: when $\frac{|\mathcal{L}|}{|T^{wm}|}$ and $|T^{wm}|$ are small, $d_{lb}=0$ with high probability (e.g., $0.5$ when $\frac{|\mathcal{L}|}{|T^{wm}|}=2$, $|T^{wm}|=1$), but this drops as they grow. A higher $\delta$ restricts synonyms, lowering $d_{ub}$ and preserving utility but reducing detectability. Conversely, a lower $\delta$ expands candidates, increasing $d_{ub}$ and detectability at the cost of greater semantic drift and reduced utility.

\textbf{(iv)} $|T^{wm}|$ correlates with text length, and larger $|T^{wm}|$ increases deviation. As shown in Fig.~\ref{fig:Theory-main-b}, $|T^{wm}|$ grows from $10$ (average 109 words) to $40$ (average 170 words). Experiments on $1,000$ samples per group show that higher $|T^{wm}|$ leads to greater deviation, reducing utility but improving detectability. For example, deviation rises from $0.008$ when $|T^{wm}|<10$ to $0.035$ when $|T^{wm}|>40$. Shorter texts (smaller $|T^{wm}|$) thus preserve utility but weaken detectability.

\section{Experiments}\label{sec:experiments} 

\subsection{Baselines, Models, Datasets,   Tasks, and Evaluation Metrics}  \label{subsec:evaluation}
% \vspace{-5pt}
% We conduct experiments on  a variety of state-of-the-art  (SOTA)  watermarks: \textbf{1)} \textit{KGW} \cite{kirchenbauer2023watermark}, which biases generation toward preferred (green) tokens; \textbf{2)} \textit{EXP} \cite{kuditipudi2023robust}, which   maps watermark keys to outputs; 
% \textbf{3)} \textit{SIR} \cite{liu2024semanticinvariantrobustwatermark}, which adjusts    logits depending on previous tokens; \textbf{4)} \textit{SynthID} \cite{dathathri2024scalable}, which uses  Tournament sampling with a pseudorandom function to select outputs; \textbf{5)} \textit{TW} \cite{yang2023watermarking}: A binary text watermark that replaces tokens with context-aware synonyms using a Bernoulli-based random encoding; 
% \textbf{6)} \textit{DeepTextMark} \cite{munyer2024deeptextmark} (referred to as \textit{DTM}), which substitutes tokens with semantically similar candidates using Word2Vec \cite{mikolov2013efficient}; and  \textbf{7)} A lexical substitution watermark \textit{LW} \cite{he2022protecting} that replaces adjectives from WordNet \cite{miller1995wordnet} and spelling variants.

Baselines include \textbf{1)} \textit{KGW} \cite{kirchenbauer2023watermark} (biased green tokens); \textbf{2)} \textit{EXP} \cite{kuditipudi2023robust} (key mapping); \textbf{3)} \textit{SIR} \cite{liu2024semanticinvariantrobustwatermark} (logits adjustment via context); \textbf{4)} \textit{SynthID} \cite{dathathri2024scalable} (tournament sampling); \textbf{5)} \textit{TW} \cite{yang2023watermarking} (binary substitution); \textbf{6)} \textit{DTM} \cite{munyer2024deeptextmark} (Word2Vec-based substitution); and \textbf{7)} \textit{LW} \cite{he2022protecting} (WordNet-based  substitution).

% To assess the robustness of watermarks, we evaluate against SOTA watermark removal and model stealing attacks:  \textbf{1)} \textit{Dipper} \cite{krishna2024paraphrasing}, which is a sentence-level attack that paraphrases the outputs through context reordering and lexical changes;  
%  \textbf{2)} \textit{Substitution} \cite{pan2024markllm}, which is a token-level attack that randomly substitutes selected tokens with synonyms from WordNet  \cite{miller1995wordnet};  \textbf{3)} \textit{SIRA} \cite{cheng2025revealing}, which leverages each token's self-information to conduct   paraphrasing attacks;
%  and \textbf{4)} \textit{Model stealing} \cite{birch2023model,li2023protecting,sander2024watermarking,dang2025delta}, where an adversary replicates an LLM's  behaviors by querying outputs and using prompt-output pairs to train a surrogate model. 
%  Baseline watermarks and related tasks  follow their original experimental settings.

To assess robustness, we evaluate against SOTA removal and model stealing attacks: \textbf{1)} \textit{Dipper} \cite{krishna2024paraphrasing} (sentence-level paraphrasing); \textbf{2)} \textit{Substitution} \cite{pan2024markllm} (token-level synonym replacement via WordNet \cite{miller1995wordnet}); \textbf{3)} \textit{SIRA} \cite{cheng2025revealing} (self-information–based paraphrasing); and \textbf{4)} model stealing \cite{birch2023model,li2023protecting,sander2024watermarking,dang2025delta} (training a surrogate from prompt–output pairs). Baselines follow their original settings.

% \subsection{Models, Datasets, and Tasks}  

% We conduct experiments using five open-source LLMs, including LLaMA-2 7B \cite{touvron2023llama2openfoundation} and Mistral 7B \cite{jiang2023mistral7b} as primary models  for all tasks, and DeepSeek 7B \cite{bi2024deepseek}, Qwen2.5 7B \cite{qwen2.5}, and Gemma 7B \cite{team2403gemma} for additional cross-provider detection.  
% We evaluate on a primary text generation  task using the C4 dataset \cite{dodge2021documentinglargewebtextcorpora} and two downstream tasks, which are  MMLU \cite{hendrycks2021measuringmassivemultitasklanguage} for multi-task understanding and CNN/Daily Mail \cite{chen2016thorough} for text summarization.

We conduct experiments on five 7B LLMs: LLaMA-2 \cite{touvron2023llama2openfoundation} and Mistral \cite{jiang2023mistral7b} as primary models, and DeepSeek \cite{bi2024deepseek}, Qwen2.5 \cite{qwen2.5}, and Gemma \cite{team2403gemma} for cross-provider detection. Evaluation includes C4 text generation \cite{dodge2021documentinglargewebtextcorpora}, MMLU Question-answering \cite{hendrycks2021measuringmassivemultitasklanguage} and CNN/Daily Mail summarization \cite{chen2016thorough}.

% \vspace{-5pt}

% \subsection{Evaluation Metrics}\label{subsec:evaluation}
% \vspace{-5pt}

For comprehensive evaluation, we use multiple metrics, including \textit{1) Utility} is evaluated using i) BERTScore \cite{zhang2020bertscore} for text semantic similarity, ii) Entity Similarity Score for factual consistency, iii) visually qualitative comparisons, and iv) accuracy on MMLU;  \textit{2) Detectability} is measured using Detection Rate and  ROC-AUC \cite{fawcett2006roc}.
\textit{3) Human Evaluation}  on Amazon Mechanical Turk (AMT) to assess text quality, including side-by-side and head-to-head comparisons, as well as the impact of attacks. 
\noindent Additional details on metric definitions and evaluation procedures are  in Appx.~\ref{app:evaluation}.

% \begin{figure}[t]
% \centering
% \subfigure[LLaMA-2]
% {\includegraphics[scale=0.245]{images/Trade_offLLaMA.png}}%\hfill
% \subfigure[Mistral]{\includegraphics[scale=0.245]{images/Trade_offMistral.png}}%\vfill
% \caption{Text generation results in C4 dataset.} \vspace{-7.5pt}
% \label{fig:LLM_WM}
% \end{figure}  

\subsection{Experimental Results}
\label{sec:results}
\vspace{-10pt}
% \setlength{\textfloatsep}{4pt}

% Our experiments  provide insights into the trade-off between LLM's utility and WM's detectability of  \textsc{SafeSeal} compared to SOTA watermarks, across LLMs, watermark detectors, and attacks, with text generation and downstream tasks. In addition, we  conduct an ablation study to assess the impact of various components of \textsc{SafeSeal} on this trade-off and provide guidelines to control it. 

% Our experiments provide insights into the trade-off between LLM's utility and watermark detectability for \textsc{SafeSeal}, in comparison to  SOTA watermarks, across different LLMs, watermark detectors, and attacks, including text generation and downstream tasks.  In addition, we conduct an ablation study to assess the impact of various components of \textsc{SafeSeal} on this trade-off and provide guidelines to control it.

We study the following questions on LLaMA-2. More results  on Mistral and others are in Appx.~\ref{app:additional_results}.

% Our experiments address the following questions on LLaMA-2; Mistral results are in Appx.~\ref{app:additional_results}.

\begin{figure}[t]
\centering
\subfigure[Text generation]
{\includegraphics[scale=0.245]{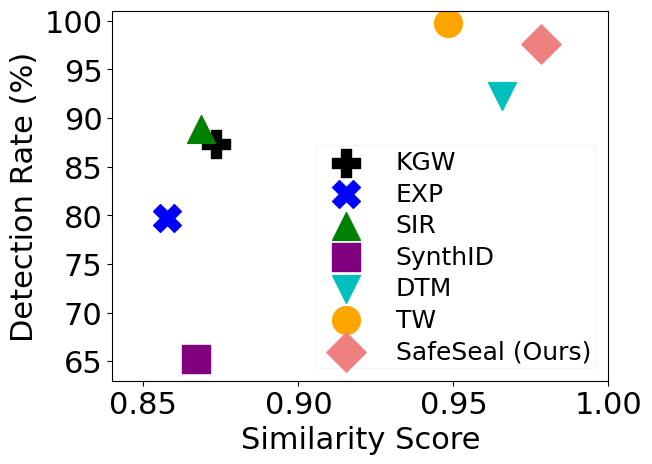}\label{fig:textgen_llama2}}\hfill
\subfigure[Text summarization]{\includegraphics[scale=0.245]{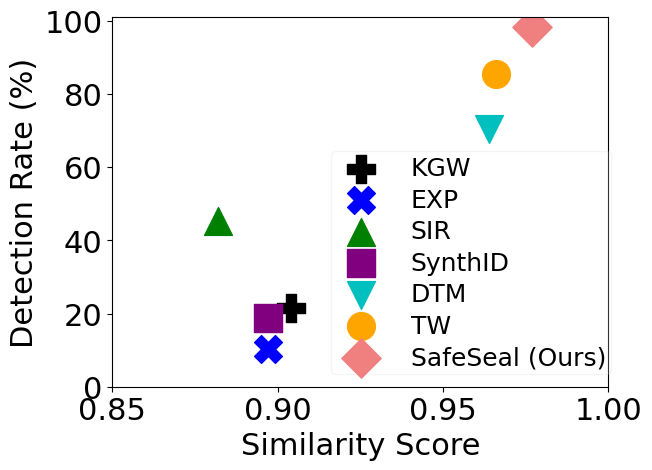}\label{fig:textsum_llama2}}\hfill
\subfigure[Entity preservation]{\includegraphics[scale=0.24]{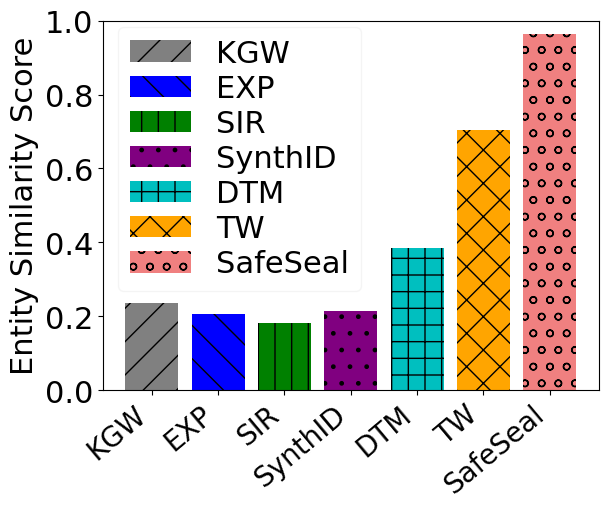}\label{fig:entity_llama2}}\vfill
\caption{Utility and detectability performance on text generation and summarization.} \vspace{-7.5pt}
\label{fig:LLM_WM}
\end{figure}  

\begin{table}[t]
\footnotesize
\centering
\caption{Stealing attack performance. (BERT is BERTScore, DR is detection rate, HR is hit rate)}
\label{tab:model_stealing_llama2}
\renewcommand{\arraystretch}{1.1}
\begin{tabular}{p{1.1cm} p{0.7cm} p{0.7cm} p{1.1cm} p{0.7cm} p{0.7cm} p{1.1cm} p{0.7cm} p{0.55cm} p{1.3cm}}
\toprule
\multicolumn{3}{c}{\textbf{\textsc{SafeSeal}}} &
\multicolumn{3}{c}{\textbf{\textsc{TW} \cite{yang2023watermarking}}} &
\multicolumn{4}{c}{\textbf{LW} \cite{he2022protecting}} \\
\cmidrule(lr){1-3} \cmidrule(lr){4-6} \cmidrule(lr){7-10}
\textbf{Setting} & \textbf{BERT}$\uparrow$ & \textbf{DR}$\uparrow$ &
\textbf{Setting} & \textbf{BERT}$\uparrow$ & \textbf{DR}$\uparrow$ &
\textbf{Setting} & \textbf{BERT}$\uparrow$ & \textbf{HR}$\uparrow$ & \textbf{p-value}$\downarrow$ \\
\midrule
Watermark & 0.977 & 98.2\% & Watermark & 0.966 & 85.5\% & Watermark & 0.985 & 0.33 & $4.1{\times}10^{-11}$ \\
Adv.(10k) & 0.867 & \textbf{68.3\%} & Adv.(10k) & 0.879 & 4.4\% & Adv.(10k) & 0.873 & 0.14 & 1.0 \\
Adv.(20k) & 0.871 & \textbf{72.8\%} & Adv.(20k) & 0.873 & 5.4\% & Adv.(20k) & 0.878 & 0.16 & 1.0 \\
\bottomrule
\end{tabular}
\end{table}

\begin{figure}[t]
\centering
\subfigure[Human evaluation on text]{
\includegraphics[scale=0.25]{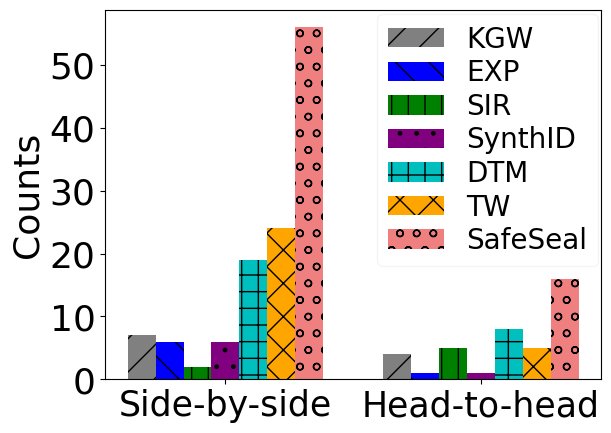}
\label{fig:ATM_text-a}
}\hfill
\subfigure[Human evaluation on removal]{
\includegraphics[scale=0.25]{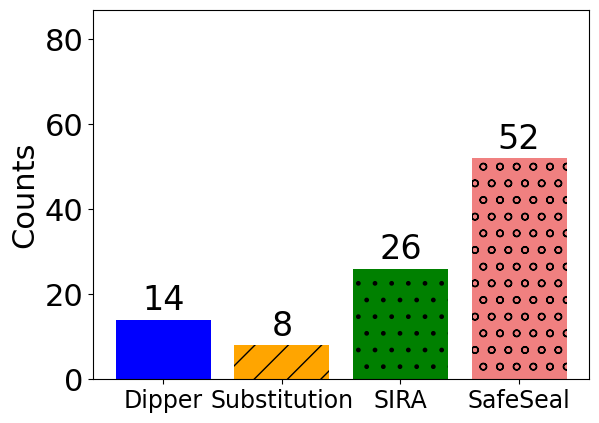}
\label{fig:WM_removal_Quality-b}
}\hfill
\subfigure[Detection on removal attacks]{
\includegraphics[scale=0.205]{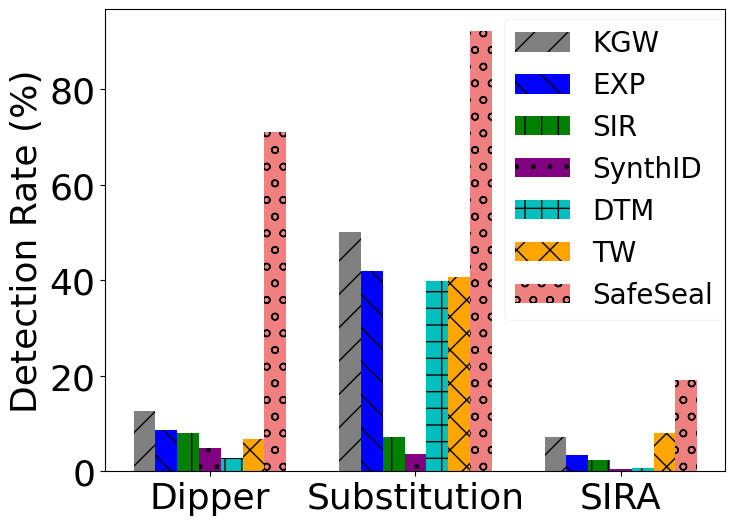}
\label{fig:WM_removal_detection}
}\vfill
\caption{AMT results on text generation, removal attacks, and detection on the attacks (LLaMA-2).}  %\vspace{-7.5pt}
\label{fig:WM_removal_Quality}
\end{figure}

\begin{table}[t]
\footnotesize
\centering
\caption{Entity preservation example. (Green  is \hl{preserved entities} and red  is \hlred{modified entities})}
% (Yellow text highlights the key information to protect, green text indicates preserved information, and red text shows modified key information)}
\begin{tabular}{|p{0.084\linewidth}|p{0.845\linewidth}|}
\hline
\textbf{} & \centering \textbf{Content} \tabularnewline \hline
\centering \textbf{Prompt} & 
\raggedright ...For the week, the Dow fell 0.8 percent, the S\&P 500 slipped 0.9 percent and the ... \tabularnewline \hline
\centering  \multirow{2
}{*}{\makecell{\textbf{No} \\ \textbf{Watermark}}}  & 
\raggedright \hl{Nasdaq} dropped \hl{0.7 percent}. The report from the \hl{Commerce Department} showed that the \hl{U.S.} economy grew at a slower-than-expected \hl{2.5 percent} annual rate in the \hl{second quarter}, down from a \hl{3.4 percent} pace in the \hl{first quarter}.[...] \tabularnewline \hline
\centering \multirow{2
}{*}{\makecell{\textbf{KGW} \cite{kirchenbauer2023watermark}}} & 
\raggedright \hl{Nasdaq} lost \hl{0.7 percent}. The economic report from the \hl{Commerce Department} showed that the \hl{U.S.} economy grew at a slower-than-expected \hlred{3.5 percent} annual rate in the \hlred{first quarter}, down from a previously estimated \hlred{4.4 percent pace}.[...]  
\tabularnewline \hline
\centering \multirow{2
}{*}{\makecell{\textbf{EXP} \cite{kuditipudi2023robust}}} & 
\raggedright \hl{Nasdaq} lost \hlred{0.6 percent}. The \hl{U.S.} economy grew at a slower-than-expected \hlred{1.9 percent} annual rate in the \hlred{first quarter}, the \hl{Commerce Department} reported. Economists had been expecting a rate of \hlred{2.2 percent}.[...] 
\tabularnewline \hline
\centering
\multirow{2}{*}{\makecell{\textbf{SIR} \cite{liu2024semanticinvariantrobustwatermark}}} & 
\raggedright \hl{Nasdaq} lost \hlred{0.6 percent}. The \hl{Commerce Department} said Friday that the \hl{U.S} economy grew at a \hl{2.5 percent} annual rate in the \hlred{first quarter}, slower than the \hlred{4.4 percent} rate in the previous quarter[...] 
\tabularnewline \hline
\centering
\multirow{2}{*}{\makecell{\textbf{SynthID} \\\cite{dathathri2024scalable} }} & 
\raggedright  \hl{Nasdaq} lost \hlred{0.4 percent}. The \hl{U.S.} \hl{Commerce Department} said Friday that the economy grew at a lackluster \hlred{1.9 percent} annual rate in the \hl{second quarter}, down from a \hlred{2.1 percent} rate in the \hl{first quarter}.[...] 
\tabularnewline \hline
\centering
\multirow{2}{*}{\makecell{\textbf{DTM} \cite{munyer2024deeptextmark}}} & 
\raggedright \hlred{nasdaq} dropped \hlred{0.7 percentage}. the report from the \hlred{commerce department} showed that the \hlred{u.s.} economy grew at a slower-than-expected \hlred{2.5 percentage} annual rate in   \hl{second quarter}, down from a \hl{3.4 percent} pace in the \hl{first quarter}.[...] 
\tabularnewline \hline
\centering
\multirow{2}{*}{\makecell{\textbf{TW} \cite{yang2023watermarking}}} & 
\raggedright \hl{Nasdaq} dropped \hlred{0 . 7 percent} . The statement from the \hl{Commerce Department} showed that the \hlred{U . S .} economy grew at a slower - than - planned \hlred{2 . 5 percent} yearly   in the \hlred{second half} , down from a \hlred{3 . 4 percent} pace in the \hl{first quarter} .[...] 
\tabularnewline \hline
\centering \multirow{2}{*}{\makecell{\textbf{\textsc{SafeSeal}} \\ \textbf{(Ours)}}} & 
\raggedright \hl{Nasdaq} fell \hl{0.7 percent}. The report from the \hl{Commerce Department} showed that the \hl{U.S.} economy grew at a slower-than-expected \hl{2.5 percent} annual rate in the \hl{second quarter}, down from a \hl{3.4 percent} pace in the \hl{first quarter}.[...] 
\tabularnewline \hline
\end{tabular}
\label{table:sensitivity}
\end{table}
\setlength{\textfloatsep}{5pt}
\setlength{\floatsep}{5pt}

 \textbf{Q1. What is the trade-off between utility and detectability of \textsc{SafeSeal} and watermarks?} 
 % Fig.~\ref{fig:textgen_llama2} compares watermark performance across LLMs and shows that \textsc{SafeSeal} achieves the best balance between model utility and watermark detectability. \textsc{SafeSeal} achieves the highest similarity score with a BERTScore of $0.978$, while the scores for other watermarks range from   $0.858$ in EXP to   $0.966$ in DTM. In terms of watermark detectability, \textsc{SafeSeal} achieves a very competitive detection rate of $98.2\%$, which is only second to TW ($99.8\%$). However, TW suffers from a markedly lower BERTScore ($0.951$) because it severely distorts content meaning, limiting its practicality despite its high detection rate. Similar results using Mistral further strengthen our observations. Also,   DTM shows a notable lower watermark detection rate ($92.3\%$) than \textsc{SafeSeal}. ROC curves (Appx.~\ref{app:additional_results}) further confirm these findings, with \textsc{SafeSeal} achieving near-perfect separability (AUC up to $0.99$) and consistently outperforming or matching strong baselines across models
 Fig.~\ref{fig:textgen_llama2} compares watermark performance and shows that \textsc{SafeSeal} achieves the best utility-detectability balance. It attains the highest BERTScore ($0.978$), while others range from $0.858$ (EXP) to $0.966$ (DTM). Its detection rate ($98.2\%$) is competitive, second only to TW ($99.8\%$), though TW sacrifices utility ($0.951$ BERTScore). DTM also shows lower detectability ($92.3\%$). Results on Mistral are consistent, and ROC curves further confirm near-perfect separability for \textsc{SafeSeal}. % (AUC up to $0.99$).
We attribute \textsc{SafeSeal}'s strong utility and detectability to four factors: 1) excluding entities and using context-aware synonyms to reduce semantic drift, 2) key-conditioned sampling for provider/user-specific signals, 3) contrastive learning with text-key feature fusion, and 4) a non-intrusive design that independently watermarks tokens without altering generation, preventing cascade effects (Appx.~\ref{app:additional_results}).
% We attribute \textsc{SafeSeal}'s strong performance in both model utility and watermark detectability to four key design characteristics: 1) separating entities from watermarking to preserve susceptible content while using context-aware synonyms to minimize semantic drift, 2) employing a key-conditioned sampling to embed provider/user-specific watermark signals, 3) incorporating contrastive learning and text-key feature fusion to enhance detection robustness and effectiveness, and 4) adopting a non-intrusive design that watermarks tokens independently without altering the LLM's generation, preventing cascade effects (Appx.~\ref{app:additional_results}).

%  \begin{figure}[t]
% \centering
% \subfigure[LLaMA-2]
% {\includegraphics[scale=0.245]{images/Summary_Task_LLaMA.png}}\hfill
% \subfigure[Mistral]{\includegraphics[scale=0.245]{images/Summary_Task_Mistral.png}}\vfill
% \caption{Text summarization downstream task results.} %\vspace{-7.5pt}
% \label{fig:Downstream_summarization}
% \end{figure} 

\textbf{Q2. How do watermarks impact downstream tasks?} Fig.~\ref{fig:textsum_llama2} shows that \textsc{SafeSeal} outperforms existing watermarks on summarization, achieving higher BERTScore and detection. It reaches $0.977$ BERTScore and $98.4\%$ detection, exceeding the best baseline by $12.9\%$. All in-processing methods, which rely on statistical rules, struggle on short texts (e.g., KGW at $21.5\%$, SIR at $45.3\%$).
On the MMLU, \textsc{SafeSeal} preserves the original model's accuracy (Table~\ref{table:MMLU}), thanks to its token-filtering watermarking step that avoids altering label tokens. Other methods reduce accuracy by $1.7\%$ - $40.2\%$ on LLaMA-2 and $0.3\%$ - $64.5\%$ on Mistral, or at best match our performance.

% In Fig.~\ref{fig:textsum_llama2}, \textsc{SafeSeal} outperforms   existing watermarks on text summarization, achieving   higher  BERTScore and stronger detection. \textsc{SafeSeal} obtains $0.977$ BERTScore, compared to $0.882$ (SIR) to $0.966$ (TW), and  $98.4\%$ detection rate, exceeding the best baseline by $12.9\%$. 
% In-processing watermarks perform poorly on short-text summarization, with SIR reaching only $45.3\%$. Similar results appear for Mistral, where \textsc{SafeSeal} achieves a $0.985$ BERTScore and $97.4\%$ detection rate. These consistent gaps further highlight \textsc{SafeSeal}'s effectiveness

\textbf{Q3. Can watermarks preserve  entities?} In Fig.~\ref{fig:entity_llama2}, \textsc{SafeSeal} achieves the highest entity similarity ($0.963$), far surpassing the best baseline (TW) at $0.704$.
Table~\ref{table:sensitivity} shows that other methods often distort content. For instance,  ``0.7 percent'' becomes ``0.6 percent'' (EXP, SIR), ``0.4 percent'' (SynthID), ``0 . 7 percent'' (TW, with token-splitting issues). Similarly, the US economy grew at a ``2.5 percent'' is shifted to ``3.5 percent'' (KGW) or ``1.9 percent'' (EXP). 
\textsc{SafeSeal}'s high entity preservation comes from excluding entities during watermarking, yielding strong similarity scores. In contrast, TW and DTM lack such filtering, and in-processing methods that  watermark at every generation step can inadvertently modify entities. Such distortions are especially problematic in domains like medicine or finance, where small changes can have serious consequences. \textsc{SafeSeal} achieves $0.963$ entity similarity, as minor drops arise when new phrases are introduced (e.g., ``at this moment''  becomes ``at this minute''), which preserve meaning but count as mismatches.

\textbf{Q4. Is \textsc{SafeSeal} effective in human evaluations?} Fig.~\ref{fig:ATM_text-a} shows that \textsc{SafeSeal} ranks highest in both AMT side-by-side and head-to-head settings, reflecting strong text quality. In the side-by-side setting, \textsc{SafeSeal} receives $56$ votes,  followed by $24$ for the best baseline (TW).  In head-to-head setting, \textsc{SafeSeal} receives  $96$ votes (normalized to $16$ for a fair comparison with other six watermarks), doubling the best baseline (DTM, $8$).  Table~\ref{table:viz_LLM2_text} further confirms this with examples where \textsc{SafeSeal} best preserves semantics, factual, and stylistic consistency.
% Fig.~\ref{fig:WM_removal_Quality-a} shows that \textsc{SafeSeal} consistently ranks highest in both AMT side-by-side and head-to-head settings, highlighting its ability to preserve high text quality after watermarking. In the side-by-side setting, \textsc{SafeSeal} receives $56$ votes, followed by $24$ for the best baseline (TW).  In the head-to-head setting, \textsc{SafeSeal} obtains $96$ votes, which normalizes to $16$ to ensure fair comparison with the other six watermarks, double that of the best baseline (DTM, $8$ votes). 
% Table~\ref{table:viz_LLM2_text} (Appx.~\ref{app:additional_results}) further illustrates this trend with a line-by-line example, where \textsc{SafeSeal} best preserves semantics, factual consistency, and stylistic consistency relative to other watermarks.

% \begin{figure}[t]
% \centering
% \subfigure[Text generation]{
% \includegraphics[scale=0.25]{images/Human_EvalLlama2.png}
% \label{fig:WM_removal_Quality-a}
% }\hfill
% \subfigure[Watermark removal attacks]{
% \includegraphics[scale=0.25]{images/WM_removal_Quality_Llama2_added.png}
% \label{fig:WM_removal_Quality-b}
% }\vfill
% \caption{AMT results.}
% \label{fig:WM_removal_Quality}
% \end{figure}

\textbf{Q5. Is \textsc{SafeSeal} robust to watermark removal attacks?} 
Fig.~\ref{fig:WM_removal_detection} shows  \textsc{SafeSeal}  outperforms other watermarks in removal attacks. It achieves   detection rates of  $71.1\%$ in Dipper, $92.2\%$ in Substitution and   $19.2\%$ in SIRA , while other watermarks drop to $2.7\%$ - $12.63\%$, $3.7\%$ - $50.1\%$ and $0.4\%$ - $8.1\%$, respectively.  Additionally, these attacks often degrade text quality, limiting the usefulness of removed outputs for surrogate training \cite{erol2025sandcastles}. While SIRA reduces detection the most, it also severely harms quality. In Fig.~\ref{fig:WM_removal_Quality-b}, $52\%$ of \textsc{SafeSeal} outputs are preferred for meaning preservation, dropping to $8\%$ - $26\%$ under attacks, which explains the reduced detection rates.

% Fig.~\ref{fig:WM_removal} shows that  \textsc{SafeSeal}  outperforms other watermarks in resisting watermark removal attacks across LLMs. 
% On  LLaMA-2,  it achieves   detection rates of  $71.1\%$ in Dipper, $92.2\%$ in Substitution and   $19.2\%$ in SIRA , while other watermarks drop to $2.7\%-12.63\%$ in  Dipper, $3.7\%-50.1\%$ in  Substitution and $0.4\%-8.1\%$ in  Substitution.  Similar trends hold for  Mistral.

% \begin{figure}[t]
% \centering
% \subfigure[LLaMA-2]{
% \includegraphics[scale=0.21]{images/WM_RemovalLlama2_added.png}
% }\hfill
% \subfigure[Mistral]{
% \includegraphics[scale=0.21]{images/WM_RemovalMistral_added.png}
% }\vfill
% \caption{Detection rates on watermark removal attacks.}
% \label{fig:WM_removal}
% \end{figure}

 % In addition, these attacks often degrade both detectability and text quality, limiting the adversary from obtaining high-quality outputs after removal to train its surrogate model \cite{erol2025sandcastles}. Notably, SIRA achieves the most significant reduction in detection rates, but this comes at the cost of substantially degraded text quality.  In Fig.~\ref{fig:WM_removal_Quality-b}, $52\%$ of \textsc{SafeSeal} outputs without attacks were preferred for  meaning preservation,  compared to  $8\%-26\%$ under attacks. This drop explains the lower detection rates, as distorted text is harder  to identify. 

\textbf{Q6. Is  \textsc{SafeSeal}   effective under model stealing attacks?} 
Model stealing experiments confirm \textsc{SafeSeal}'s robustness. In Table~\ref{tab:model_stealing_llama2}, a surrogate, trained on \textsc{SafeSeal}-watermarked summaries to replicate  the service provider's behaviors  in this downstream task, achieves $72.8\%$ detection and $0.871$ BERTScore, compared to $98.4\%$ and $0.977$ without attacks. TW attains similar BERTScore but very low detection ($5.4\%$), while LW shows low hit rates ($\leq 0.33$), far below the $0.81$ threshold \cite{he2022protecting}, and becomes undetectable under attack ($p=1$). In-processing methods (e.g., KGW, SynthID) perform poorly on short texts ($18.7\%$ - $21.5\%$) and drop to near $0\%$ under stealing (Table~\ref{tab:stealing_attack_in_processing}).

\begin{figure}[t]
\centering
\subfigure[Watermark Leaderboard]{
\includegraphics[scale=0.238]{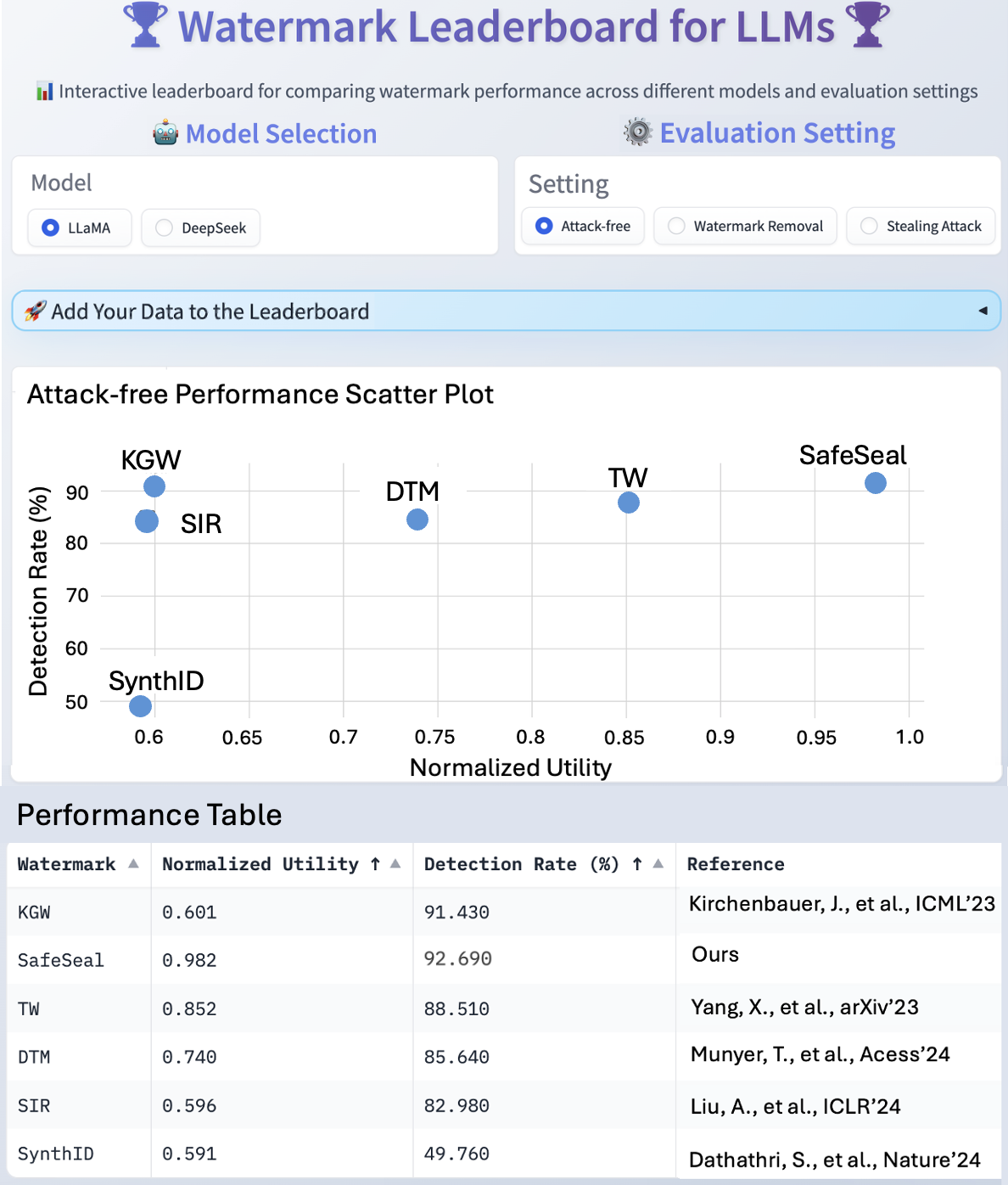}
\label{fig:leaderboard}}\hfill
\subfigure[\textsc{SafeSeal} demo interface.]{
\includegraphics[scale=0.28]{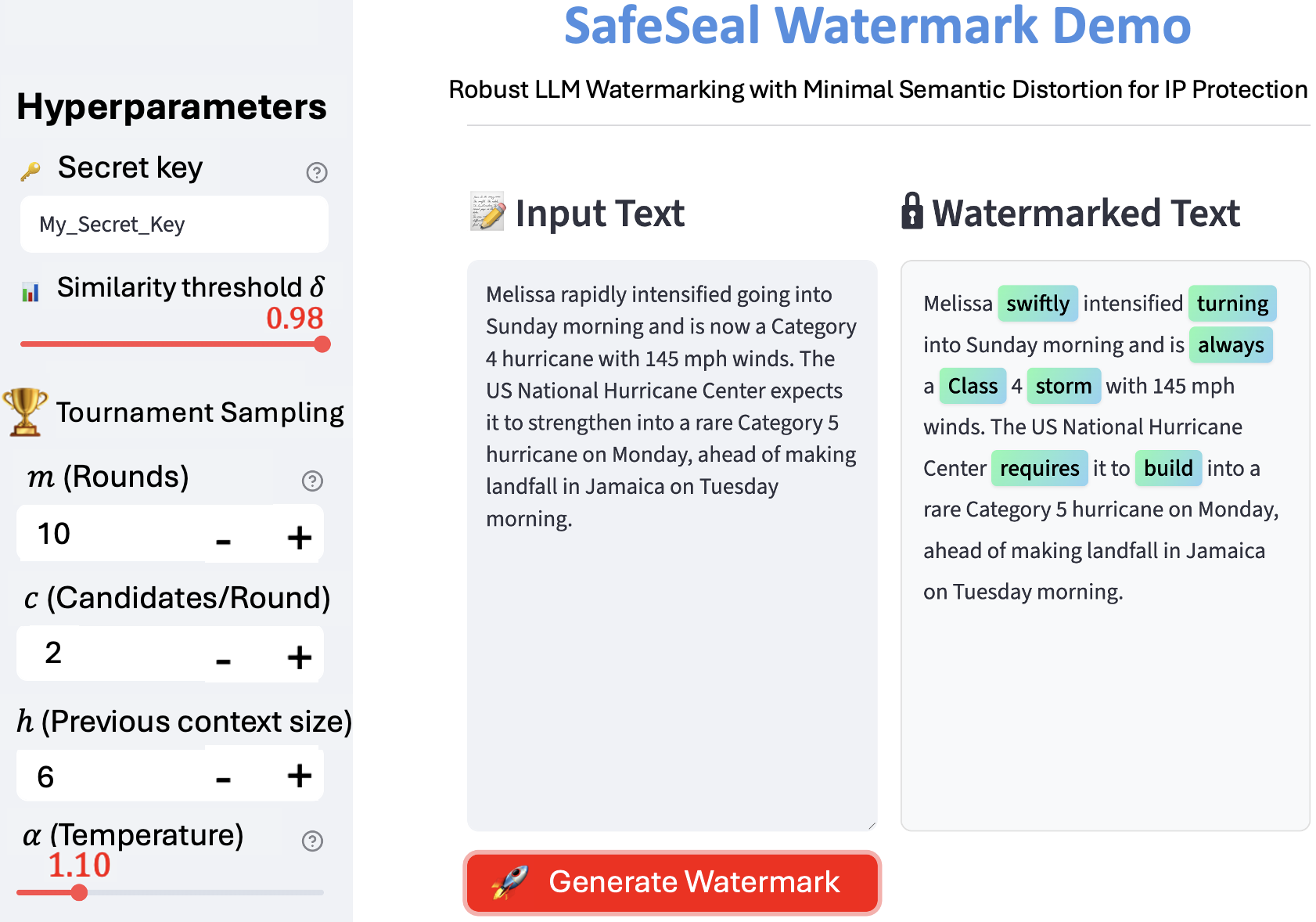}
\label{fig:interface}}\vfill
\caption{Watermark leaderboard and interactive demo interface.}
%\vspace{-7.5pt}
\label{fig:combine_HF}
\end{figure}  

\begin{table}[t]
\footnotesize
\centering
\caption{Cross-provider, multi-user, and latency results.}
\label{tab:combined_all}

\begin{minipage}[t]{0.35\linewidth}
\centering
\textbf{(a) Cross-provider}\\[-0.2em]
\begin{tabular}{p{1.0cm} p{1.2cm} p{0.5cm} p{0.5cm}}
\toprule
Method & LLaMA-2 & Mistral & DeepSeek \\
\midrule
KGW              & 93.2 & \textbf{98.1} & \textbf{99.3} \\
EXP              & 79.8 & 76.3 & 89.9 \\
SIR              & 86.8 & 50.1 & 49.0 \\
SynthID          & 50.5 & 50.7 & 52.1 \\
DTM              & 91.8 & 9.8  & 10.1 \\
TW               & 93.9 & 8.2  & 14.0 \\
\textsc{SafeSeal}& \textbf{98.2} & 92.7 & 98.1 \\
\bottomrule
\end{tabular}
\label{tab:cross_provider}
\end{minipage}
\hfill
\begin{minipage}[t]{0.2\linewidth}
\centering
\textbf{(b) Multi-key}\\[-0.2em]
\begin{tabular}{p{0.4cm} p{0.7cm} p{0.6cm}}
\toprule
Keys & Cross-provider & Multi-user \\
\midrule
2 & 95.20 & 88.79 \\
3 & 94.75 & 90.35 \\
4 & 94.42 & 91.38 \\
5 & 95.48 & 91.63 \\
\bottomrule
\end{tabular}
\label{tab:multi_key}
\end{minipage}
\hfill
\begin{minipage}[t]{0.33\linewidth}
\centering
\textbf{(c) Latency (s)}\\[-0.2em]
\begin{tabular}{p{1.0cm} p{0.9cm} p{1.1cm}}
\toprule
Method & Mean $\downarrow$ & Median $\downarrow$ \\
\midrule
NW          & 4.71  & 5.03 \\
KGW         & +13.71 & +14.75 \\
EXP         & +297.28 & +296.95 \\
SIR         & +3.88  & +4.40 \\
SynthID     & +0.23  & +0.18 \\
DTM         & +4.47  & +4.61 \\
TW          & +1.97  & +1.98 \\
\textsc{SafeSeal} & +3.42  & +3.36 \\
\bottomrule
\end{tabular}
\label{tab:latency}
\end{minipage}

\end{table}

\textbf{Q7. Is the detector effective under real-world scenarios?} 
 \textit{Cross-provider (single-key training).} Trained on LLaMA-2 with a single key, our detector generalizes across providers, achieving $98.2\%$ on LLaMA-2 and high non-watermarked accuracy on Mistral ($92.7\%$) and DeepSeek ($98.1\%$) (Table~\ref{tab:cross_provider}(a)). In contrast, SIR and SynthID remain near random (approximately $50\%$), while DTM and TW fail to  generalize.  \textit{Cross-provider and Multi-user (multi-key training).} A shared detector trained with multiple keys maintains high detectability across both settings, achieving $>94\%$ accuracy for cross-provider (up to $95.48\%$ with five LLM providers) and $\sim92.1\%$ for multi-user on a same LLM (LLaMA-2) without retraining (Table~\ref{tab:multi_key}(b)), reducing deployment complexity and computational cost.
 \textit{Different data distributions.} 
In Table~\ref{tab:model_stealing_llama2},  a detector trained on $20000$ C4 samples generalizes  well to CNN/DailyMail summarization, highlighting the effectiveness of our detector.

\textbf{Q8. How  is \textsc{SafeSeal}'s latency?}  
Table~\ref{tab:latency}(c) shows \textsc{SafeSeal} has modest latency ($3.42$s), much lower than most baselines. KGW is $\sim5\times$ slower due to repeated green-list updates; EXP is $\sim90\times$ slower from full-vocabulary scoring; DTM and SIR are $1.30\times$ and $1.14\times$ slower, respectively. TW and SynthID have low overhead ($0.23$–$1.97$s) but poor utility and robustness. \textsc{SafeSeal}'s overhead primarily stems from candidate generation; however, it remains practical and further optimizable.

% Table~\ref{tab:latency}c) shows \textsc{SafeSeal} has a modest latency of $3.42$s, which is much lower than those of other watermarks. For instance, KGW is about $5\times$  slower due to repeated green-list construction and per-token probability updates; EXP is roughly  $90\times$ slower due to modified sampling and exponential scoring over the full vocabulary; DTM is $1.30\times$ slower; and SIR is  $1.14\times$ slower due to semantic embedding computations over prior tokens. Differently, TW and SynthID exhibit minor overhead ($0.23-1.97$s) but suffer from poor  utility and robustness across tasks, as shown in previous  experiments. The overhead of \textsc{SafeSeal} primarily comes from its context-aware candidate generation and filtering to preserve utility; however,  the overall cost remains practical and can be further optimized.

\textbf{Q9. What community-driven tools can enhance LLM watermarking development?} We introduce two tools.  \textit{First}, the leaderboard (Fig.~\ref{fig:leaderboard}) provides standardized benchmarking under consistent settings. Users can select LLMs and settings, with results shown via plots and tables. Methods are evaluated for clear trade-off comparison, while supporting community submissions for reproducibility. \textit{Second}, the SafeSeal demo (Fig.~\ref{fig:interface}) enables real-time experimentation. Users adjust key parameters and observe their effects, with inputs and highlighted watermarked outputs shown side-by-side.

\section{Conclusion and Future Work}
\label{sec:conclusion}

This work presents \textsc{SafeSeal}, a key-conditioned LLM watermark that balances detectability and utility with provable guarantees. By preserving entities and replacing linguistic terms with context-aware synonyms, \textsc{SafeSeal} maintains semantics and factual consistency. Key-conditioned Tournament sampling produces reproducible provider-specific signals, while a contrastive detector improves robustness across LLMs and users without retraining. \textsc{SafeSeal} also improves efficiency through lightweight components, caching, batching, and framework consolidation. We further release the first public watermark leaderboard and an interactive demo for standardized evaluation. Future work includes further improving the watermark-naturalness trade-off, especially in low-diversity domains; developing stronger probabilistic guarantees beyond deviation-based analysis; and  reducing latency through optimized candidate generation and integration with frameworks such as vLLM \cite{kwon2023efficient}.

\newpage
\clearpage

{
\small
\bibliographystyle{plainnat}
\bibliography{SafeSeal}
}

\newpage
\clearpage

%%%%%%%%%%%%%%%%%%%%%%%%%%%%%%%%%%%%%%%%%%%%%%%%%%%%%%%%%%%%

\appendix

\section{Background and Related Work  }
\label{sec:related_work}

\subsection{Preliminaries and Key Definitions}
\textbf{Large Language Models (LLMs).} Given a vocabulary $\mathcal{V}$ of words or word fragments, known as \textit{tokens}, a sequence  $x = \{x^{(i)}\}_{i=1}^T \in \mathcal{V}^T$ is a \textit{prompt} of   $T$  tokens, and $y =  \{y^{(i)}\}_{i=T+1}^{T+Q}\in \mathcal{V}^Q$ is the   output of $Q$ tokens  ($T,Q > 0$). An LLM  $\theta$ is trained to maximize a conditional probability  $P(y|x)= \Pi_{i=1}^T P(y^{(i)}|x^{(1)},  \cdots, x^{(i-1)})$, capturing long-range contextual dependencies for   generation.

\textbf{Semantic Meaning and Factual Consistency.} Semantic meaning captures contextual coherence and intended meaning \cite{Devlin2019BERTPO}, while factual consistency requires that watermarked text remain faithful to the facts and entities in the original content \cite{nan-etal-2021-entity}.  
Distortions occur when substitutions alter either meaning or facts, making both essential for high-quality watermarking.

 \textbf{Named Entity Recognition (NER) and Part-of-Speech (POS) Tagging.}  NER \cite{Nadeau2007ASO} identifies   entities such as names, locations, and time expressions. Each entity $e \in \mathcal{E}$ is a word or a phrase treated as highly susceptible, as minor changes can alter  semantic meaning, factual consistency, and compromise downstream tasks.   POS tagging \cite{toutanova2003feature} assigns each word a grammatical category (e.g., nouns,  verbs). Widely used NER and POS tagging tools are spaCy \cite{spaCy2020}, Stanza \cite{stanza2020}, and NLTK \cite{bird2009nltk}. 

 % \textbf{Named Entity Recognition (NER) and Part-of-Speech (POS) Tagging.} NER \cite{Nadeau2007ASO} identifies entities such as person names,   locations, and time expressions. Each entity $e \in \mathcal{E}$  is a word or a phrase   treated as highly susceptible, since   minor changes can alter meaning and harm downstream tasks. Preserving entities in  watermarking is   essential for factual consistency.
 
% \textbf{Part-of-Speech (POS) Tagging.} POS tagging \cite{toutanova2003feature} assigns each word a grammatical category (e.g., nouns, verbs, adjectives, and adverbs). In this work, we use these substitutable categories for context-aware synonym substitution. Widely used NER and POS tagging tools include spaCy \cite{spaCy2020}, Stanza \cite{stanza2020}, and NLTK \cite{bird2009nltk}. 

 \textbf{Cross-Provider and Multi-User Detection.} 
 In practice, different providers or users may use distinct secret keys. We consider two settings: \textit{(1) cross-provider detection}, where different LLM providers use different keys, and \textit{(2) multi-user detection}, where one provider serves multiple users, each with a unique key. In both cases,    the detector must   identify the correct source of the watermarked text.

 \textbf{Tournament Sampling \cite{dathathri2024scalable}.} This sampling uses cryptographically keyed pseudorandom scores over sliding contexts to select tokens, enabling scalable and low-latency generation. However, pseudorandom scoring can weaken semantic preservation and robustness, especially for short inputs. %\textsc{SafeSeal} addresses this by combining context-aware substitution with tournament-style key-conditioned scoring and a contrastive detector, improving utility and detection under attacks.

\subsection{Watermarking Techniques in LLMs}
% \vspace{-10pt}

Watermarks embed imperceptible patterns in LLM outputs to verify origin and deter model stealing. They include 1) a generation function $\mathcal{W}$ producing a watermarked output $y^{wm}=\theta(x,\mathcal{W})$, given the LLM $\theta$ and $x$, and 2) a detector $\mathcal{D}$ that identifies if a given text contains the watermark.

\textbf{Watermark Categories.} Existing watermarks fall into \textbf{\textit{(1) in-processing watermarks}}, which embed signals during generation by modifying logits,  token distributions, or training   \cite{kirchenbauer2023watermark, zhao2024provable, liu2024adaptive}; and \textbf{\textit{(2) post-processing watermarks}}, which modify generated text through synonym substitution,  rule-based replacements, or paraphrasing \cite{dathathri2024scalable, giboulot2024watermax,hou2023semstamp}. Detection typically uses   statistical tests or classifiers. However, both categories may distort semantic meaning  or factual consistency by altering entities. %\textsc{SafeSeal} addresses this by preserving entities and replacing terms with context-aware synonyms.
 
\begin{table}[H]
\footnotesize
\caption{Example of distortion-free watermark outputs (EXP and SynthID) on LLaMA-2. Red indicates \hlred{semantic changes}; green indicates \hl{preserved meaning}.}
\centering
\begin{tabular}{|p{0.3\linewidth}|p{0.31\linewidth}|p{0.3\linewidth}|}
\hline
\textbf{No Watermark} & \textbf{EXP} \cite{kuditipudi2023robust}  & \textbf{SynthID}  \cite{dathathri2024scalable} \\
\hline
... \hl{Scotland}, where I spent a month as a \hl{writer-in-residence.}
I've been to \hlred{New York City} \hlred{three times}, but I've never been to the \hlred{Metropolitan Museum of Art}.
I was in \hlred{New York City} \hlred{last week}, and I finally got to visit the \hlred{Metropolitan Museum of Art}. 
& 
... \hl{Scotland}, for a \hl{writer's residency.} 
\hlred{7 DAYS THAT WILL CHANGE YOUR LIFE} 
\hlred{Every year}, \hlred{hundreds of thousands of people} from all over the world come to experience the \hlred{7 Days That Will Change Your Life} program on the shores of \hlred{Lake Tahoe, California}. 
& 
... \hl{Scotland} for a month-long \hl{writing residency.}
Living in \hlred{Quezon City} in the \hlred{1970s}.
\hlred{The first time} I lived in \hlred{Quezon City} was in \hlred{1975}, when I transferred to \hlred{Don Bosco Technical College} from \hlred{Don Bosco Academy} in \hlred{Bacolor, Pampanga}. \\
\hline
...   have much \hlred{higher rates than other countries}.
In the \hl{United States}, the \hl{rate of infant mortality} \hlred{varies by state}, with \hlred{some states} having rates that are \hlred{significantly} \hl{higher} than \hlred{the national average}. 
& 
... range from \hlred{60 to over 100 deaths per 1,000}   compared to the \hl{U.S. rate} of \hlred{5.89} \hl{deaths per 1,000 births} in \hlred{2015}, \hl{up} \hlred{slightly from the year before}. \hlred{Term Papers on Demand Custom Term Papers Online.} 
& 
... are plagued by \hlred{large scale poverty, poor sanitation, and lack of access to health care}.
The \hl{U.S rate} in \hlred{2015} was \hlred{5.87} \hl{infant deaths per 1000 live births}, which is \hl{better than}    \hlred{6.18} in \hlred{2005}. \\
\hline
\end{tabular}
\label{tab:non_distortion}
\end{table}
\setlength{\textfloatsep}{10pt}
% \vspace{-20pt}
 
\textbf{Distortion-Free and Semantic Preservation.} 
Prior work  \cite{dathathri2024scalable,kuditipudi2023robust,hou2024k,christ2024undetectable,pang2025modelshield,golowich2024edit,mao2024watermark,wang2025watermarkslanguagemodelsprobabilistic}
defines distortion-free watermarking as preserving the   next-token distribution, but this does not guarantee semantic or factual consistency.  
In Table \ref{tab:non_distortion}, distortion-free methods, e.g., EXP \cite{kuditipudi2023robust} and SynthID \cite{dathathri2024scalable}  still cause   entity shifts (e.g., from a New York story  to a California retreat or a Pampanga story) and fabricated details (e.g., from U.S. state-level infant mortality rate shifts to country-level statistics over time), due to cascading effects in autoregressive generation (Appx.~\ref{app:additional_results}).% \textsc{SafeSeal} mitigates this by preserving entities and using context-aware synonym substitution in post-processing.

\textbf{Conditional Watermarking for Provider-Specific Detection.} 
Prior studies typically use secret keys to   bias token sampling \cite{kirchenbauer2023watermark,kuditipudi2023robust,hu2024unbiased,christ2024undetectable,yang2023watermarking,wang2025morphmark}, or embed provider-specific triggers  \cite{xu2024hufu,zhang2024emmark,he2022cater}, but often  assume single-key, attack-free settings with limited empirical detection analysis. In practice, cross-provider and multi-user scenarios require robust detection under unique private keys. %\textsc{SafeSeal} addresses this with a key-conditioned contrastive detector that supports reliable cross-provider and multi-user detection  while improving robustness against attacks.

 \textbf{Utility-Detectability Trade-off.} 
Prior work studies this trade-off through watermark strength optimization  \cite{kirchenbauer2023watermark, hou2024token, wang2025morphmark}, Pareto objectives \cite{optimizing2024wouters24a}, entropy-based hybrids and key randomization \cite{huang2024waterpool,wang2025trade}, and edit-robust guarantees \cite{golowich2024edit}. However, these methods often rely on proxy metrics such as probability or entropy, which may miss semantic and factual distortions.% We define \textit{utility} as preserving \textit{semantic meaning, factual consistency, and downstream task performance}.

Table \ref{tab:Summary} compares \textsc{SafeSeal} with related work. \textsc{SafeSeal} is a key-conditioned post-processing watermark that preserves entities and replaces linguistic terms with context-aware synonyms to maintain \textit{utility}, defined as  semantic meaning, factual consistency, and downstream utility. Its contrastive detector captures key-specific patterns for robust  cross-provider and multi-user detection, while our evaluation covers utility, robustness,  and the theoretical utility-detectability trade-off.

\section{\textsc{SafeSeal} LLM Watermarking} \label{appendix:Safeseal}

\subsection{Linguistic Categories}
\label{appendix:Linguistic}
 
\begin{table}[t]
\footnotesize
\caption{Linguistic categories and subcategories used in \textsc{SafeSeal}, with corresponding examples.}
\centering
\begin{tabular}{|p{0.10\linewidth}|p{0.3\linewidth}|p{0.2\linewidth}|}
\hline
\textbf{Category} & \textbf{Subcategory (Tag)} & \textbf{Example} \\
\hline
Noun & NN (Singular noun) & dog, car, house \\
     & NNS (Plural noun) & cats, people \\
\hline
Verb & VB (Base form) & run, sell \\
     & VBD (Past tense) & ran, sold \\
     & VBG (Gerund/participle) & running, selling \\
     & VBN (Past participle) & eaten, dropped \\
     & VBP (Non-3rd person singular) & eat, lie \\
     & VBZ (3rd person singular) & eats, lies \\
\hline
Adjective & JJ (Adjective) & happy, sad \\
          & JJR (Comparative) & happier \\
          & JJS (Superlative) & happiest \\
\hline
Adverb & RB (Adverb) & quickly \\
       & RBR (Comparative adverb) & faster, slower \\
       & RBS (Superlative adverb) & fastest, slowest \\
\hline
\end{tabular}
\label{tab:pos_categories}
\end{table}

Table~\ref{tab:pos_categories} presents a list of part-of-speech (POS) categories used in \textsc{SafeSeal},  focusing on  nouns, verbs, adjectives, and adverbs. We select these categories for their semantic flexibility, where synonym substitutions typically preserve overall meaning without compromising factual consistency. 
  For instance, in the noun category, we include common singular and plural nouns, while excluding proper nouns, as replacing the proper nouns (e.g., ``Google'') may change factual information. Similarly, we utilize different types of verbs, adjectives, and adverbs for watermarking, while excluding modal verbs, pronouns, determiners, numerals identified through POS tagging. By restricting substitutions to this selected set of POS categories, \textsc{SafeSeal} minimizes semantic distortion while preserving effectiveness.

% \subsection{Tournament Sampling} \label{appendix:tournament_sampling}
% \begin{figure}[t]
%       \centering    
%       \includegraphics[scale=0.065]{images/Tournament_sampling.png} %\vspace{-10pt}
%       \caption{Tournament sampling overview.} %\vspace{-10pt}
%       \label{fig:tournament_sampling}
%  \end{figure}

\subsection{Pseudo Code}
The pseudo code of \textsc{SafeSeal} is presented in Algorithm~\ref{alg:watermark_ldp}. It outlines the end-to-end watermarking pipeline, including candidate generation, key-conditioned sampling, and substitution. For clarity, we abstract implementation details and focus on the core steps involved in embedding the watermark.

\begin{algorithm}[t]
\small
\caption{\textsc{SafeSeal} Watermarking Technique }   
\label{alg:watermark_ldp}
\begin{algorithmic}[1]
\STATE \textbf{Inputs}: prompt $x$, testing prompt $x_{test}$, LLM $\theta$, $NER(\cdot)$,  similarity threshold $\delta$, lightweight LM $\mathcal{A}$, Batch size $B$, watermark detector $\mathcal{D}$ %with $f_{enc}$, $f_{MLP}$ and, $f_{out}$
\STATE \textbf{Outputs}: $y^{wm} = \theta(x, \mathcal{W})$, $ \mathcal{D}\big(\theta (x_{test},\mathcal{W})\big)$
\STATE \textbf{Watermark Generation Process}:
\begin{ALC@g}
\STATE Obtain original output: $y=\theta(x)$ 
\STATE \textbf{\textit{Entity Recognition:}} Get a set of entities: $\mathcal{E} = NER(y)$    
\STATE \textbf{\textit{Linguistic Term:}}  Get watermarkable tokens by POS tagging:\\  ~
$T^{wm} = \{t|  t= \{ \text{nouns},  \text{verbs},  \text{adjectives},   \text{adverbs} \} \in y  \text{ and } t \not\in \mathcal{E}  \}$      
\STATE \textbf{\textit{Lookup Table:}} Generate context-aware synonyms %for   watermarkable tokens
\end{ALC@g}
\begin{ALC@g}
\FOR{each watermarkable token $t_i \in T^{wm}$}
   \STATE  $\bullet$  \textit{Substitution Candidate Generation:}
    \STATE Identify context $C$ and sentence $S_{t_i}$ containing $t_i$
    \STATE Mask $t_i$ in $C$ to create the masked version $C_M$  
    \STATE Form an input for $\mathcal{A}$: $m = C \oplus [SEP] \oplus C_{M}$
    \STATE Use $\mathcal{A}$ to generate a list of   candidates   $\mathcal{L}_{t_{i}}$  
    \STATE Refine $\mathcal{L}_{t_{i}}$ (remove antonyms and incomplete words) %  by removing antonyms and incomplete words
\STATE $\bullet$ \textit{Substitution Candidate Selection:}
\STATE Replace $t_i$ with $t_{j} \in \mathcal{L}_{t_{i}}$ to get modified sentence $S_{t_j}$  
\STATE Get similarity scores for all modified sentences: $\mathbb{S}_{ij}$  %.l,of each candidate $\tilde{t_{i}}$ with \( t_i \) by $ \mathbb{S}(\tilde{t_{i}}, t_{i}) = \mathbb{S}(\tilde{S_{t_{i}}}, S_{t_{i}})$
\STATE  Rank the tokens in $\mathcal{L}_{t_{i}}$ in descending similarity order
\STATE 
Select candidates with similarity scores $\ge\delta$, ensuring at least two remain:
\\ \quad $   \mathcal{L}_{t_i}^{\delta} =  \big\{ t_j \colon \mathbb{S}_{ij} \ge \delta \big\} \text{ and } |\mathcal{L}_{t_i}^{\delta}| \geq 2 $   

\ENDFOR
\STATE Form a lookup table: $\mathcal{L} = \big\{\mathcal{L}_{t_i}^{\delta} \big\}_{t_i \in T^{wm} }$
\end{ALC@g}
\STATE \textbf{\textit{Watermark Injection via Tournament Sampling:}}
\begin{ALC@g}
\FOR{ each watermarkable token $t_i \in T^{wm}$ }
\STATE In parallel, replace each $t_i$ with the winning candidate $t_j \in \mathcal{L}_{t_i}^{\delta}$ selected by Tournament sampling.
\ENDFOR
\STATE \textbf{Return}: Watermarked output $y^{wm} = \theta(x, \mathcal{W})$.
\end{ALC@g}

\STATE \textbf{Watermark Detection Process}:
\begin{ALC@g}
\STATE Construct training pairs $\{(y_i,k_i, 0), (y_i^{wm},k_i, 1)\}_{i=1}^B$
\STATE Encode text embeddings $\mathbf{h}_y$, key embeddings $\mathbf{h}_k$, and create feature fusion: $\boldsymbol{\phi} =  \big[\,\mathbf{h}_y ; \mathbf{h}_k ; \mathbf{h}_y \odot \mathbf{h}_k ; |\mathbf{h}_y-\mathbf{h}_k|\,\big]$
\STATE Compute logits $z$ and detection loss $\mathcal{L}_{\mathrm{det}}$ (BCE)
\STATE Compute contrastive loss $\mathcal{L}_{\mathrm{con}}$ (symmetric InfoNCE)
\STATE Optimize detector by minimizing total loss:
\\ \quad $\mathcal{L} = \mathcal{L}_{\mathrm{det}} + \lambda\,\mathcal{L}_{\mathrm{con}}, \ \lambda\!\in\![0,1]$
\STATE \textbf{Return}: A trained detector $ \mathcal{D}\big(\theta (x_{test},\mathcal{W})\big)$.
\end{ALC@g}
\end{algorithmic}
\end{algorithm}

\subsection{Latency Optimization} \label{subsec:latency}
To reduce latency, \textsc{SafeSeal} uses lightweight models and several system-level optimizations. \textit{First}, caching avoids redundant NER, POS-tagging, and synonym-generation calls by reusing analyzed segments, tags, and candidates. \textit{Second}, batching and GPU parallelization accelerate candidate generation, similarity scoring, and Tournament sampling. \textit{Third}, we unify components in a PyTorch-based framework to reduce framework-switching overhead and improve throughput. Together, these optimizations make \textsc{SafeSeal} efficient for large-scale deployment.

\section{ Proof of Theorem \ref{theorem:bounds}}
\label{app:theorem}

\begin{proof}
We determine the upper and lower bounds for LLM's utility and    detectability  by finding the range for the deviation $\Delta(y, y^{wm})$.  
  However, determining this range with a closed-form solution is challenging since  language is highly nuanced and the impact of watermarks can vary across different outputs. Therefore, we approximate the range by considering the expected value of the deviation, as follows:  
\begin{equation} 
\small
\mathbb{E}\Big(\Delta(y, y^{wm})\Big) = \frac{1}{|T^{wm}|}\sum_{t_i\in T^{wm}, t_j\in \mathcal{L}_{t_i}^{\delta}}\mathbb{E}\Big(\mathbb{D}(t_i, t_j)\Big),
\label{eq:e_output}
\end{equation}
where the expected distance $\mathbb{E}\Big(\mathbb{D}(t_i, t_j)\Big)$ is computed as:
\begin{equation}
\small
    \mathbb{E}\Big(\mathbb{D}(t_i, t_j)\Big) = \frac{1}{|T^{wm}|} \sum_{t_i\in T^{wm}, t_j \in \mathcal{L}_{t_i}^{\delta},  t_i  \neq 
 t_j} \mathbb{D}(t_i, t_j) \cdot q^{t_j}_{t_i},
    \label{eq:e_distance}
\end{equation}
where $q^{t_j}_{t_i}$ is the probability of replacing $t_i$ with $t_j \in \mathcal{L}_{t_i}^{\delta}$ and $t_j \neq t_i$. If $t_j = t_i$, the distance is $0$; therefore, we only consider cases where $t_j \neq t_i$.  Since candidates are different from the original watermarkable token (the lookup table $\mathcal{L}_{t_i}^{\delta}$ contains only synonym candidates distinct from $t_i$), the probability of replacing $t_i$ with $t_j \in \mathcal{L}_{t_i}^{\delta}$ and $t_j \neq t_i$ is $1$. Assuming that every distance $\mathbb{D}(t_i, t_j)$ is bounded by a lower bound $d_{lb}$ and an upper bound $d_{ub}$, from (\ref{eq:e_distance}),  we have $\mathbb{E}\Big(\mathbb{D}(t_i, t_j)\Big) = \frac{1}{|T^{wm}|} \sum_{t_i\in T^{wm}, t_j \in \mathcal{L}_{t_i}^{\delta},  t_i  \neq 
 t_j} \mathbb{D}(t_i, t_j) \cdot q^{t_j}_{t_i} = \frac{1}{|T^{wm}|}  \sum_{i=1}^{|T^{wm}|} \mathbb{D}(t_i, t_j)$. Considering the lower bound $d_{lb}$ and the upper bound $d_{ub}$, we derive the following:
\begin{equation}
\small
 \frac{d_{lb} }{|T^{wm}|}   \leq \mathbb{E}\Big(\mathbb{D}(t_i, t_j)\Big) \leq  \frac{ d_{ub}}{|T^{wm}|}.
    \label{eq:bound}
\end{equation}
\noindent   Then,   lower and upper bounds of the expected deviation are: 
  {\small
\begin{align}
    \frac{d_{lb}}{|T^{wm}|}    \leq \mathbb{E}\Big(\Delta(y, y^{wm})\Big) \leq  \frac{d_{ub}}{|T^{wm}|}   
\end{align}
}
\noindent Consequently, Theorem \ref{theorem:bounds} holds.
\end{proof}

\section{Detailed Experimental Settings and Additional Results}
\label{app:additional_results}

Our experiments were conducted on Python version 3.9.20, using NVIDIA A100 GPU with PyTorch (torch 2.5.1) and CUDA 12.

\subsection{Task Setup} \label{append:Task_setup}
\textit{For text generation}, we use $20,000$ training and $1,000$ test samples, truncating each to $200$ prompt tokens and generating up to $200$ tokens \cite{kirchenbauer2023watermark,kuditipudi2023robust}. We use spaCy \cite{spaCy2020} for NER and NLTK \cite{bird2009nltk} for POS tagging. For lookup tables, RoBERTa-base (125M) \cite{liu2019roberta} generates candidates and DistilBERT (67M) \cite{Sanh2019DistilBERTAD} computes similarity with $\delta=0.92$. Sampling uses $\alpha=1.0$, $c=2$, $h=6$, and $m=10$.

\textit{For MMLU downstream task}, we evaluate multiple-choice question-answering across $57$ subjects. 

\textit{For text summarization}, we test $1,000$ samples with outputs averaging $60$ words. 

\textit{For detection}, we train a RoBERTa-large model \cite{liu2019roberta} (356M) with batch size $16$, $15$ epochs, learning rate $5e{-6}$, and contrastive weight $\lambda=0.15$ on $20,000$ positive and $20,000$ negative pairs.

\textit{For watermark removal attack}, we evaluate performance on $1,000$ test samples.

\textit{For the model stealing attack}, we adopt a task-specific threat model in which the adversary targets a text summarization service, a realistic scenario reflecting practical API abuse. The adversary queries $20,000$ prompts to collect \textsc{SafeSeal}-watermarked summaries and fine-tunes a surrogate model on the collected input-output pairs. We then evaluate the surrogate on $1,000$ held-out prompts against both watermarked and clean outputs to assess detectability and utility degradation.

\textit{For real-world settings}, we test 1) cross-provider detection with five LLMs using distinct keys, and 2) multi-user detection with one LLaMA-2 model serving five users with separate keys.

\textit{For latency}, we report the average latency in seconds (s) for $1,000$ generations of $200$-token outputs with LLaMA-2.

\subsection{Evaluation Metrics}\label{app:evaluation}
\noindent \textbf{Utility Preservation Metrics.}  To thoroughly evaluate  model utility, we use  a wide range of  metrics:

\noindent \quad $\bullet$ \textbf{\textit{Similarity Score:}} We use BERTScore \cite{zhang2020bertscore} to measures semantic similarity between  original and watermarked outputs as it leverages contextual embeddings and aligns with human judgments. Higher scores indicate better utility.

\noindent \quad $\bullet$ \textbf{\textit{Entity Similarity Score:}}
To evaluate entity preservation, a naive approach is to use the classical exact-match accuracy by comparing entities extracted from the original output $y$ and the watermarked output $y^{wm}$. A match requires identical entity values and does not count synonyms, such as  ``percent'' and ``percentage'' are treated as unmatched. The score is $\frac{|M_{ow}|}{|\mathcal{E}_o \cup \mathcal{E}w|}$, where $\mathcal{E}_o$ and $\mathcal{E}_w$ are the entity sets   extracted from $y$ and $y^{wm}$, and  $M_{ow}$ is the set of exact matches.  However, this metric is overly strict and insufficient to capture semantic similarity. 
To better capture this,  we use a relaxed metric  combining cosine similarity $\mathbb{S}{cos}$ \cite{singhal2001modern} (insensitive to text length but may miss subtle changes) and Levenshtein similarity $\mathbb{S}{lev}$ \cite{levenshtein1966binary} (captures subtle changes but sensitive to length). For each entity $e_o \in \mathcal{E}_o$, we apply greedy pairwise matching to find an exact match $e_w \in \mathcal{E}_w$. The final score is the average of both metrics across matched pairs, as follows:

{\small
\begin{align}
 \text{Entity Similarity Score}    = \frac{1}{|\mathcal{E}_o \cup \mathcal{E}_w|
    } \sum_{(e_o, e_w) \in M_{ow}} \frac{1}{2} \big(\mathbb{S}_{cos}(e_o, e_w) + \mathbb{S}_{lev}(e_o, e_w) \big),
    \label{eq:entitySim}
\end{align}
}

\noindent \quad $\bullet$ \textbf{\textit{Qualitative Evaluation:}}
We visually compare original and watermarked outputs across watermarking methods.

\noindent \quad $\bullet$ \textbf{\textit{Downstream Tasks:}}
 We use average accuracy on MMLU and BERTScore for text summarization, measuring similarity between watermarked and original outputs.

\noindent  \textbf{Watermark Detectability Metric.} We measure a detector's ability to distinguish watermarked from non-watermarked text using \textit{Detection Rate}, defined as follows:
{\small
\begin{equation}
\small
    \frac{ \sum_{i=1}^{N} \mathbb{I} \Big(\mathcal{D} \big( y_i^{wm} \big) = 1 \Big) + \sum_{j=1}^{M} \mathbb{I} \Big(\mathcal{D} \big(y_j \big) = 0 \Big) }{N + M},
    \label{eq:IPcheckerEq}
\end{equation}
}
\noindent in which this  metric evaluates the accuracy of a watermark detector $\mathcal{D}(\cdot)$ on $N$ watermarked and $M$ non-watermarked test samples. It uses the indicator function $\mathbb{I}$ to count correct classifications: $\mathbb{I}(\mathcal{D}(y_i^{wm})=1)$ for watermarked outputs and $\mathbb{I}(\mathcal{D}(y_j)=0)$ for non-watermarked outputs. 

It is important to note that each baseline watermark has its own detector (typically statistical or model-based). For a comprehensive and fair evaluation, we report the \textit{maximum detection rate} across each watermark's original detector and our contrastive detector, meaning baselines benefit from whichever detector performs better for them, while \textsc{SafeSeal} is evaluated with its own detector only. This conservative reporting gives baselines every advantage in detectability, making \textsc{SafeSeal}'s consistently strong results a lower bound on its true performance.

\noindent  \textbf{Human Evaluations.} These evaluations offer a holistic assessment of context, tone, readability, relevance, and overall naturalness, while effectively detecting issues like unnatural phrasing or semantic drift \cite{celikyilmaz2020evaluation, spataru2024know}. To strengthen our evaluation, we conduct human evaluations via Amazon Mechanical Turk (AMT).

\noindent \quad $\bullet$ \textbf{\textit{Text Generation:}} We use \textit{side-by-side} and \textit{head-to-head} comparisons. In the side-by-side, workers choose the watermarked output closest to the original in semantics and grammar, with watermark names hidden, positions shuffled, and responses averaged across three workers per example ($120$ examples in total). In the head-to-head, \textsc{SafeSeal} is directly compared to an arbitrary method from the baseline watermarks using $120$ examples, with workers selecting the better output relative to the original LLM output.

\noindent \quad $\bullet$ \textbf{\textit{Attacks:}}  Workers assess the impact of  removal attacks on text quality by comparing \textsc{SafeSeal} outputs, their attacked versions, and the original LLM outputs.

\subsection{ROC Curves of Watermark Detection}

\begin{figure}[H]
\centering
\subfigure[LLaMA-2]
{\includegraphics[scale=0.25]{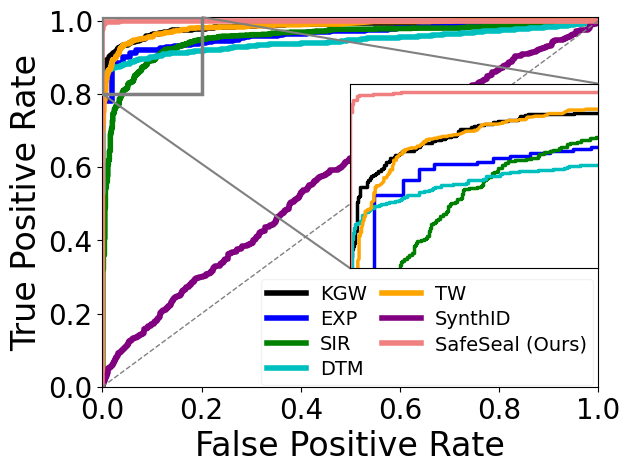}}\hspace{2cm}
\subfigure[Mistral]{\includegraphics[scale=0.25]{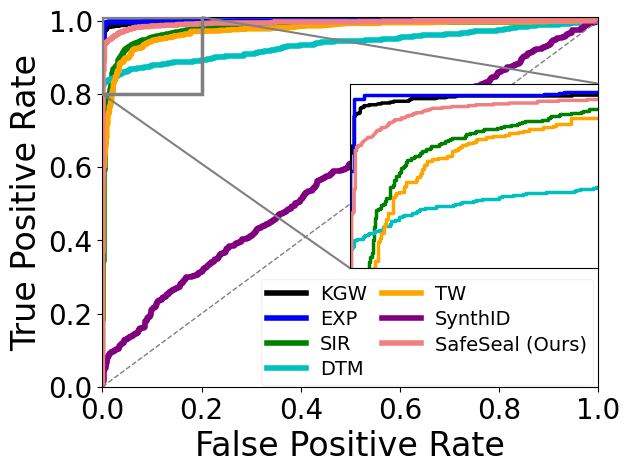}}\vfill
\caption{ROC curves for watermark detection in text generation under attack-free settings.} \vspace{-7.5pt}
\label{fig:ROC}
\end{figure}

In the attack-free setting, Fig.~\ref{fig:ROC} shows that \textsc{SafeSeal} delivers consistently strong detection performance across models. On LLaMA-2, it achieves near-perfect separability (AUC = 0.9999), outperforming all baselines. On Mistral, it remains highly competitive with an AUC of 0.9937, closely matching the best-performing methods. By contrast, SynthID performs much worse on both models (AUC $\approx$ 0.59), suggesting limited discriminative capability in this $200$-token generation setting. Overall, these results highlight \textsc{SafeSeal} as a highly reliable watermarking method, with strong and stable detection performance across different backbones.

\subsection{Text Generation and Summarization on Mistral}

\begin{figure}[t]
\centering
\subfigure[Text generation]
{\includegraphics[scale=0.245]{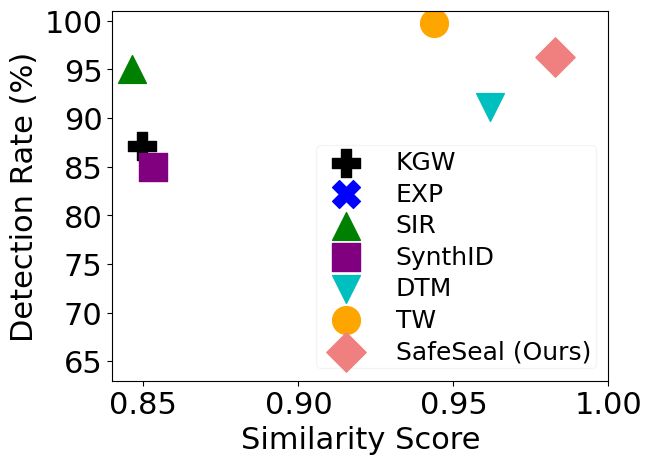}\label{fig:textgen_mistral}}\hfill
\subfigure[Text summarization]{\includegraphics[scale=0.245]{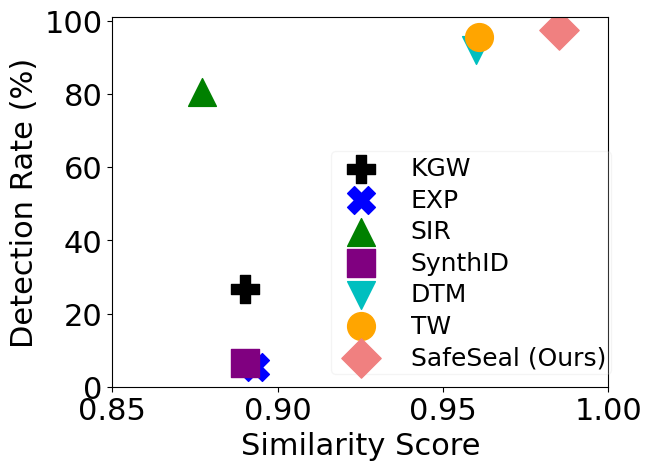}\label{fig:textsum_mistral}}\hfill
\subfigure[Entity preservation]{\includegraphics[scale=0.24]{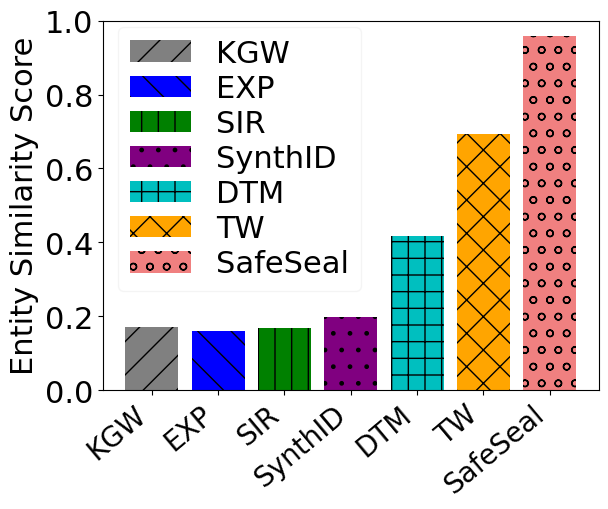}\label{fig:entity_mistral}}\vfill
\caption{Performance of text generation and summarization on Mistral.} %\vspace{-7.5pt}
\label{fig:LLM_WM_mistral}
\end{figure} 

Similar trends are observed on Mistral across text generation (Fig.~\ref{fig:textgen_mistral}), summarization (Fig.~\ref{fig:textsum_mistral}), and entity preservation (Fig.~\ref{fig:entity_mistral}). \textsc{SafeSeal} achieves the best utility–detectability balance, with the highest BERTScore ($0.983$) compared to $0.832$–$0.962$ for baselines, and a competitive detection rate ($96.25\%$), second only to TW ($99.8\%$), which sacrifices utility ($0.944$). On summarization, it attains $0.985$ BERTScore and $97.4\%$ detection. It also achieves the highest entity similarity ($0.9588$), far exceeding the best baseline (TW, $0.692$).

 \begin{figure}[t]
\centering
% \subfigure[LLaMA-2]
\includegraphics[scale=0.26]{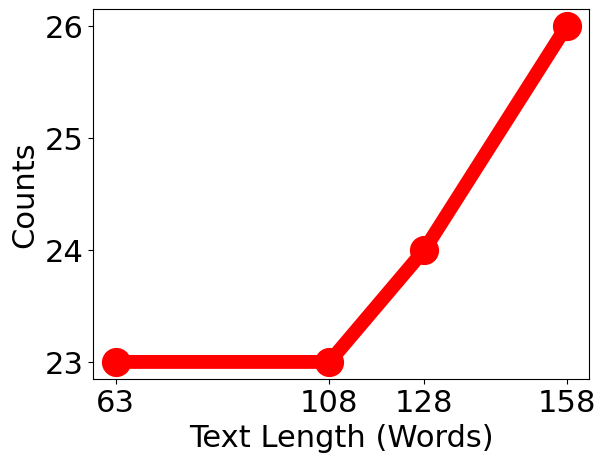}\hfill
% \subfigure[Mistral]{\includegraphics[scale=0.22]{samples/images/WM_RemovalMistral.png}}\vfill
\caption{Preference of \textsc{SafeSeal} over other watermarks in head-to-head human evaluation with different text length.} %\vspace{-7.5pt}
\label{fig:length_choice}
\end{figure}  
 
\subsection{Cascade Effect in Watermarking}
 The cascade effect refers to a phenomenon where modifying a single token in LLM's generation phase can influence the prediction of subsequent tokens. Since LLMs generate text auto-regressively, each token is predicted based on the previous ones, thus even a small change can propagate through the sequence. When a watermark is injected  in an intrusive manner, as in the case of in-processing watermarking, it modifies the generation of subsequent tokens, causing the output to diverge from the original text. This   effect becomes particularly significant when the number of generated tokens increases. 
 
 In Fig.~\ref{fig:length_choice}, as   text length increases, \textsc{SafeSeal} gains more favorable votes as the best watermark from AMT workers. For instance, at a text length of approximately $63$ words, \textsc{SafeSeal} was selected $23$ times, which increased to $26$ selections when the length reached $158$ words. As a post-processing  and non-intrusive watermarking method, \textsc{SafeSeal} can avoid the cascade effect and  more effectively preserve model utility, particularly in long-form text.

\begin{table}[H]
\footnotesize
    \centering
    \caption{MMLU accuracy (\%) and relative changes (\%) compared to non-watermarked outputs.}
    \label{table:MMLU}
    \renewcommand{\arraystretch}{1}
    \begin{tabular}{
        p{2.3cm} 
        >{\centering\arraybackslash}p{1.0cm} 
        >{\centering\arraybackslash}p{0.95cm} 
        >{\centering\arraybackslash}p{1.1cm} 
        >{\centering\arraybackslash}p{0.95cm}}  
        \toprule
        \multicolumn{1}{c}{} & \multicolumn{2}{c}{\textbf{LLaMA-2}} & \multicolumn{2}{c}{\textbf{Mistral}} \\
        \cmidrule(lr){2-3} \cmidrule(lr){4-5}
        \textbf{Setting} & \textbf{Accuracy} & \textbf{Change} & \textbf{Accuracy} & \textbf{Change} \\
        \midrule
        No Watermark              & 46.7 & --   & 58.9 & --    \\
        KGW \cite{kirchenbauer2023watermark}                & 44.5 & -4.7 & 58.6 & -0.5  \\
        EXP \cite{kuditipudi2023robust}               & 45.9 & -1.7 & 58.7 & -0.3  \\
        SIR \cite{liu2024semanticinvariantrobustwatermark}               & 27.9 & -40.2 & 20.9 & -64.5 \\
        SynthID \cite{dathathri2024scalable}             & \textbf{46.7} & \textbf{0.0} & \textbf{58.9} & \textbf{0.0} \\
        DTM \cite{munyer2024deeptextmark}                 & \textbf{46.7} & \textbf{0.0} & \textbf{58.9} & \textbf{0.0} \\
        TW \cite{yang2023watermarking}                & \textbf{46.7} & \textbf{0.0} & \textbf{58.9} & \textbf{0.0} \\
        \textbf{\textsc{SafeSeal (Ours)}} & \textbf{46.7} & \textbf{0.0} & \textbf{58.9} & \textbf{0.0} \\
        \bottomrule
    \end{tabular}
\end{table}
 
\subsection{Performance on Downstream Task} In addition to text summarization, we also conduct experiments on MMLU tasks. In Table \ref{table:MMLU},  \textsc{SafeSeal} preserves the original model accuracy without
any degradation. Meanwhile, other
watermarks reduce accuracy by $1.7\%$ to $40.2\%$ on LLaMA-2 and $0.3\%$ to $64.5\%$ on Mistral, or at best achieve comparable results to ours (e.g., SynthID, DTM, and TW). 
 The accuracy preservation in \textsc{SafeSeal}, as well as in DTM and TW, results from their non-intrusive nature and token filtering steps for watermarking, which mitigate changes in short and constrained outputs, such as single-letter multiple-choice labels (i.e., A, B, C, D, E). In the case of SynthID, as noted in \cite{dathathri2024scalable}, short text generation almost always returns the exact same output, reducing the effectiveness of the Tournament sampling process. As a result, the final output remains identical to the original model's prediction, preserving $100\%$  of the baseline performance.

\begin{table}[t!]
 \footnotesize
\centering
\caption{An example prompt and corresponding outputs from the LLaMA-2 model, shown without and with different watermarking methods. (Green text highlights words that are semantically similar to the original output)} 
\begin{tabular}{|p{0.09\linewidth}|p{0.765\linewidth}|p{0.04\linewidth}|}
\hline
\textbf{} & \centering \textbf{Output text} & \textbf{Score} \tabularnewline \hline
\centering \textbf{Prompt} & 
\raggedright ...medical laboratory technician associate program. FILING SUIT: Attorney Scott Schutzman discusses a class-action lawsuit against Newbridge College with two of his clients. A group of students... & --
\centering  \tabularnewline \hline
\centering  
\multirow{3}{*}{\makecell{\textbf{No} \\ \textbf{Watermark}}}  & 
\raggedright are suing the college, claiming its medical laboratory technician program is not accredited and does not qualify them to work in the field. SANTA ANA – A group of students have filed a class-action lawsuit against Newbridge College, claiming the college's medical laboratory technician program is [...continues] & --
\centering \multirow{3}{*}{}  \tabularnewline \hline
\centering \multirow{3
}{*}{\makecell{\textbf{KGW} \cite{kirchenbauer2023watermark}  }} & 
\raggedright \hl{are suing the college} for allegedly misrepresenting the quality of its \hl{medical laboratory technician program}. \hl{Newbridge College} faces \hl{lawsuit} over alleged false advertising of \hl{medical lab technician program} \hl{SANTA ANA} – \hl{A group of students} have \hl{filed a class-action lawsuit} \hl{against Newbridge College}, alleging the \hl{college} misrepresented [...continues] & 
\centering \multirow{4}{*}{\makecell{0.915}}  \tabularnewline \hline
\centering \multirow{3
}{*}{\makecell{\textbf{EXP} \cite{kuditipudi2023robust} }} & 
\raggedright sued \hl{the college} claims the school's program does not \hl{qualify} them to be \hl{medical lab technicians}. STUDENTS: ATTORNEY SCOTT SCHUTZMAN: Two of the \hl{Newbridge College} students who \hl{filed a class-action lawsuit} against the school are pictured with their attorney, Scott Schutzman. [...continues] & 
\centering \multirow{3}{*}{\makecell{0.904}}  \tabularnewline \hline
\centering
\multirow{4}{*}{\makecell{\textbf{SIR} \cite{liu2024semanticinvariantrobustwatermark}}} & 
\raggedright sued \hl{the college} for false advertising after they say the college's \hl{medical laboratory technologist program does not prepare them for the job market}. CENTER: \hl{Newbridge College}'s Santa Ana campus is located at 2840 East 17Th Street. Students who attended the college have \hl{filed a class- action lawsuit} against the college alleging false advertising.\hl{SANTA ANA} - Students who enrolled in \hl{Newbridge College} [...continues] & 

\centering \multirow{3}{*}{\makecell{0.881}}  \tabularnewline \hline
\centering
\multirow{3}{*}{\makecell{\textbf{SynthID} \cite{dathathri2024scalable} }} & 
\raggedright \hl{are suing the college} for false advertising and fraud after claiming the school's \hl{medical lab technician program} \hl{does not qualify them for the job market}. At a glance: \hl{Class-action lawsuit} filed \hl{against Newbridge College}.
*Plaintiffs: A group of students.
*Allegations: False advertising and fraud
[...continues] & 

\centering \multirow{4}{*}{\makecell{0.884}}  \tabularnewline \hline
\centering
\multirow{3}{*}{\makecell{\textbf{DTM} \cite{munyer2024deeptextmark}}} & 
\raggedright \hl{are suing the college}, \hl{claiming its medical laboratory technician program is not accredited and does not qualify them to working in the field}. santa ana - a \hl{group of students have filed a class-action lawsuit against} newbridge college, \hl{claiming the college's medical laboratory technician program is} [...continues] & 
\centering \multirow{4}{*}{\makecell{0.965}}  \tabularnewline \hline
\centering
\multirow{3}{*}{\makecell{\textbf{TW} \cite{yang2023watermarking}}} & 
\raggedright \hl{are suing the college} , arguing \hl{its medical laboratory technician program is not accredited and does not qualify them to work in the field} . \hl{SANTA ANA} - \hl{A group of} graduates \hl{have filed a} class - action \hl{lawsuit against Newbridge College} , arguing the college ' s \hl{medical laboratory technician program is} [...continues] & 
\centering \multirow{3}{*}{\makecell{0.962}}  \tabularnewline \hline
\centering \multirow{3
}{*}{\makecell{\textbf{\textsc{SafeSeal}} \\ \textbf{(Ours)}}} & 
\raggedright \hl{are} contacting \hl{the institution, asserting its medical laboratory technician program is not accredited and does not qualify them to work in the industry}. \hl{SANTA ANA - A handful of students have registered a class-action litigation against Newbridge College}, \hl{claiming the college's medical research technician program is}  [...continues] & 
\centering \multirow{4}{*}{\makecell{\textbf{0.968}}}  \tabularnewline \hline
\end{tabular}
\label{table:viz_LLM2_text}
\end{table}

\subsection{Qualitative Evaluation of Watermarked Text}
Table~\ref{table:viz_LLM2_text} shows an example with a prompt, the original LLaMA-2 output, and outputs from different watermarks.  Among them, \textsc{SafeSeal} effectively preserves the semantics   while introducing subtle linguistic variations.    For instance,  in  \textsc{SafeSeal}, ``the institution'' replaces ``the college,'' adding formality while maintaining semantics.  Similarly, ``work in the industry'' becomes ``work in the field,'' maintaining the context.  Also, ``a group of students'' becomes ``a handful of students,'' both referring to multiple individuals while preserving the plural subject reference.  In contrast, other watermarks often introduce semantic drift, grammatical errors, and unnecessary tokens. For instance, SIR  adds new phrases, such as ``Santa Ana campus is located at 2840 East 17Th Street,'' while  EXP significantly modifies the original meaning to ``STUDENTS: ATTORNEY SCOTT SCHUTZMAN: Two of the Newbridge College students who filed a class-action lawsuit,'' which raises concerns about factual preservation. DTM exhibits grammatical errors, such as ``to working in the field,'' and  lowercases all words, including proper nouns like ``santa ana'' and ``newbridge college.'' TW typically inserts unnecessary  stop-words and spaces. The similarity scores   further confirm the observations, with \textsc{SafeSeal} achieving the highest score of $0.968$. These results  highlight \textsc{SafeSeal}'s effectiveness in generating high-quality and semantically preserving watermarked text.

 % \textbf{Varying the lookup table size $K$ in the Mistral.} We observe similar trends in the Mistral model. In Fig.~\ref{fig:T_Length_Mistral}a, as  $K$ increases, the expected deviation of  \textsc{SafeSeal} also rises from $0.014$ to $0.022$, yet still remaining lower than that of other watermarks. In Fig.~\ref{fig:Bound}b, higher values of $K$ improve   detection rates but at the cost of utility, reinforcing the trade-off between watermark detectability and   utility.

\begin{figure}[t]
\centering
\subfigure[Deviation across $\delta$ values]{\includegraphics[scale=0.24]{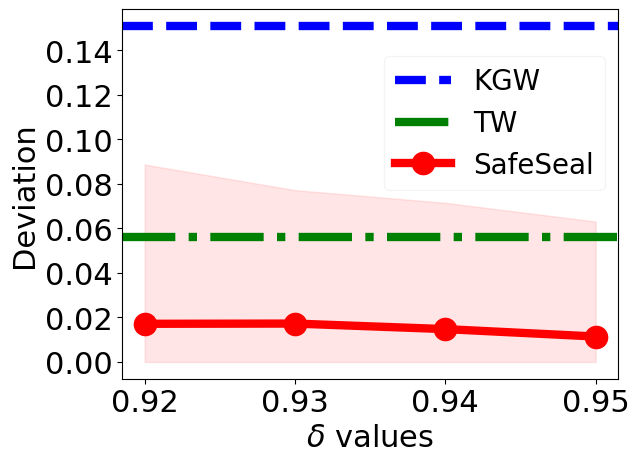} \label{fig:bound_delta_mistral}} \hspace{2cm}
\subfigure[Deviation, text length across $|T^{wm}|$]{\includegraphics[scale=0.24]{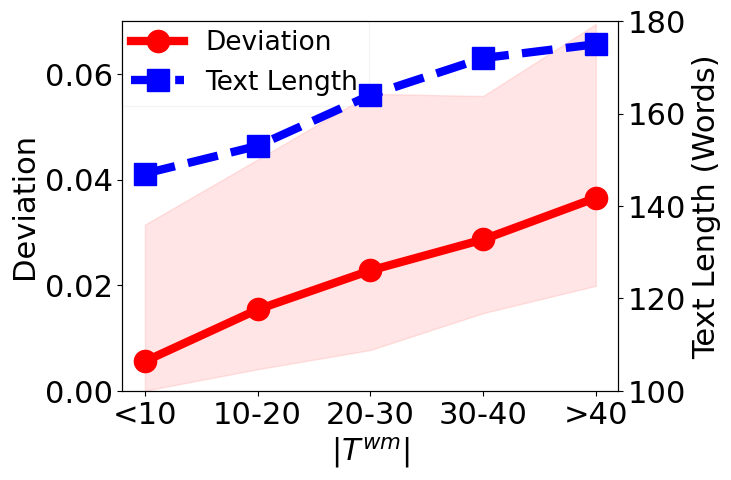} \label{fig:bound_text_length_mistral}}
\caption{Impact of lookup similarity threshold $\delta$, watermarkable set size $|T^{wm}|$, and  correlation with text length for Mistral. } %\vspace{-7.5pt}
\label{fig:T_Length_Mistral}
\end{figure}

\subsection{Utility and Detectability Bounds on Mistral} We observe consistent trends on Mistral, reinforcing our findings across LLMs. Fig.~\ref{fig:bound_delta_mistral} shows that \textsc{SafeSeal} achieves a better utility–detectability trade-off with lower expected deviations across $\delta$. Fig.~\ref{fig:T_Length_Mistral} further shows that $|T^{wm}|$ correlates with text length, and larger $|T^{wm}|$ leads to higher deviation between original and watermarked outputs.

\begin{table}[t]
\footnotesize
\centering
\caption{Stealing attacks on Mistral. (BERT is BERTScore, DR is detection rate, HR is hit rate)}
\label{tab:model_stealing_mistral}
\renewcommand{\arraystretch}{1.1}
\begin{tabular}{p{1.1cm} p{0.7cm} p{0.7cm} p{1.1cm} p{0.7cm} p{0.7cm} p{1.1cm} p{0.7cm} p{0.55cm} p{1.1cm}}
\toprule
\multicolumn{3}{c}{\textbf{\textsc{SafeSeal}}} &
\multicolumn{3}{c}{\textbf{\textsc{TW} \cite{yang2023watermarking}}} &
\multicolumn{4}{c}{\textbf{LW} \cite{he2022protecting}} \\
\cmidrule(lr){1-3} \cmidrule(lr){4-6} \cmidrule(lr){7-10}
\textbf{Setting} & \textbf{BERT}$\uparrow$ & \textbf{DR}$\uparrow$ &
\textbf{Setting} & \textbf{BERT}$\uparrow$ & \textbf{DR}$\uparrow$ &
\textbf{Setting} & \textbf{BERT}$\uparrow$ & \textbf{HR}$\uparrow$ & \textbf{p-value}$\downarrow$ \\
\midrule
Watermark & 0.985 & 97.4\% & Watermark & 0.961 & 95.6\% & Watermark & 0.982 & 0.28 & $2.3{\times}10^{-4}$ \\
Adv.(10k) & 0.878 & \textbf{50.0\%} & Adv.(10k) & 0.876 & 0.5\% & Adv.(10k) & 0.866 & 0.01 & 1.0 \\
Adv.(20k) & 0.880 & \textbf{61.9\%} & Adv.(20k) & 0.879 & 1.0\% & Adv.(20k) & 0.877 & 0.11 & 1.0 \\
\bottomrule
\end{tabular}
\end{table}

\begin{table}[t]
\footnotesize
    \centering
    \caption{Additional model stealing attacks across LLMs. (DR is detection rate.)}
    \label{tab:stealing_attack_in_processing}
    \renewcommand{\arraystretch}{1.1}
    \begin{tabular}{p{1.5cm} p{1.3cm} p{1.3cm} p{1.3cm} p{1.3cm} p{1.3cm} p{1.3cm}}
        \toprule
        \multicolumn{1}{c}{} & \multicolumn{3}{c}{\textbf{LLaMA-2}} & \multicolumn{3}{c}{\textbf{Mistral}} \\
        \cmidrule(lr){2-4} \cmidrule(lr){5-7}
         \textbf{Setting} & \textbf{BERTScore}$\uparrow$ & \multicolumn{2}{c}{\textbf{Detection Rate}$\uparrow$} & \textbf{BERTScore}$\uparrow$ & \multicolumn{2}{c}{\textbf{Detection Rate}$\uparrow$} \\
        \midrule
        \textbf{\textsc{KGW} }    & 0.904 & \multicolumn{2}{c}{21.5\%} & 0.890 & \multicolumn{2}{c}{26.6\%} \\
        Adv.(10k)        & 0.868 & \multicolumn{2}{c}{0.1\%} & 0.876 & \multicolumn{2}{c}{0.5\%} \\
        Adv.(20k)       & 0.877 & \multicolumn{2}{c}{0.0\%} & 0.879 & \multicolumn{2}{c}{1.0\%} \\
        \midrule
        \textbf{\textsc{SynthID} }    & 0.897 & \multicolumn{2}{c}{18.7\%} & 0.890 & \multicolumn{2}{c}{6.4\%} \\
        Adv.(10k)        & 0.872 & \multicolumn{2}{c}{0.0\%} & 0.881 & \multicolumn{2}{c}{0.1\%} \\
        Adv.(20k)       & 0.876 & \multicolumn{2}{c}{0.2\%} & 0.883 & \multicolumn{2}{c}{0.4\%} \\
        \bottomrule
    \end{tabular}
\end{table}

\subsection{Model Stealing Attacks} 
Additional results on model stealing attacks for Mistral with TW \cite{yang2023watermarking} and LW \cite{he2022protecting}   (Table~\ref{tab:model_stealing_mistral}) and for both LLMs for KGW \cite{kirchenbauer2023watermark} and SynthID \cite{dathathri2024scalable}(Table~\ref{tab:stealing_attack_in_processing}) show consistent trends. For post-processing methods (TW \cite{yang2023watermarking}, LW \cite{he2022protecting}), \textsc{SafeSeal} consistently outperforms baselines. Similarly, in-processing watermarks (KGW \cite{kirchenbauer2023watermark}, SynthID \cite{dathathri2024scalable}) also perform poorly under model stealing: already weak on short-text tasks, their detection drops to near $0$ – $1\%$ when a surrogate reproduces outputs. This further highlights the robustness of \textsc{SafeSeal} across both LLMs.

\begin{figure}[t]
\centering
  \centering
  \subfigure[Detection across detectors]{
  \includegraphics[scale=0.30]{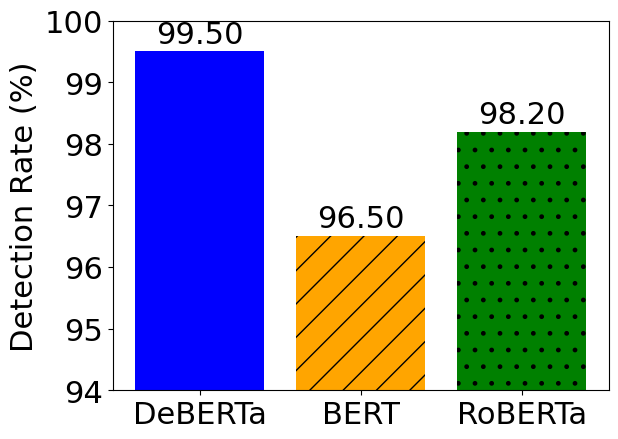}
  \label{fig:Varying_IPChecker}}\hspace{2cm}
  \subfigure[Varying similarity threshold $\delta$]{
  \includegraphics[scale=0.29]{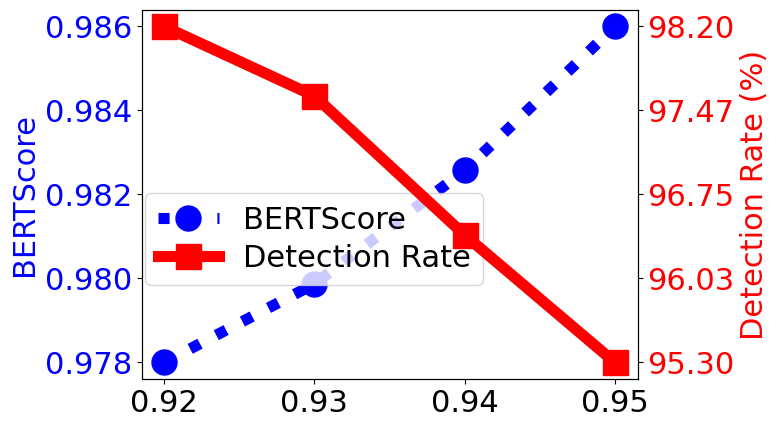}
  \label{fig:Varying_delta}}\vfill
  \caption{Detection rates across detectors.}
\end{figure}

\subsection{Justification for Q7. Cross-provider Performance}
Table~\ref{tab:cross_provider}(a) shows that our detector trained on LLaMA-2 can generalize well on other LLMs like Mistral and DeepSeek. We attribute this generalization to two complementary factors. \textit{First,} the contrastive detector is trained on matched and mismatched text-key pairs, so outputs from other providers, which carry no \textsc{SafeSeal} key-aligned substitution signal, are naturally classified as non-watermarked. \textit{Second,} different LLMs exhibit distinct output distributions and generation styles, which further separates their outputs from \textsc{SafeSeal}-watermarked text and strengthens detection, rather than hindering it.

\subsection{Additional Experimental Questions}
\textbf{Q10. How does the choice of watermark detector affect \textsc{SafeSeal}'s detectability? } 
We evaluate the impact on detectability using  three LMs of varying sizes, including  DeBERTa-large-v3 \cite{he2021deberta} (304M), BERT-large-uncase \cite{Devlin2019BERTPO} (336M), and RoBERTa-large \cite{liu2019roberta} (356M). Across the LMs, detectability remains high, with a detection rate exceeding $96\%$, with DeBERTa achieving $99.5\%$ despite its smallest size (Fig.~\ref{fig:Varying_IPChecker}). Overall, our key-conditioned contrastive detector is robust to the choice of LM backbone, demonstrating strong generality and effectiveness.

\textbf{Q11. How does the similarity threshold $\delta$ affect the utility-detectability trade-off? }  
Larger similarity thresholds $\delta$ preserve higher utility but reduce detectability (Fig.~\ref{fig:Varying_delta}). For instance, with $\delta = 0.95$, the model maintains high utility (BERTScore $= 0.986$) but obtains a lower detection rate of $95.3\%$. Reducing $\delta$ to $0.92$ slightly decreases utility (BERTScore $= 0.978$) while yielding a notable gain in detectability, with the detection rate improving to $98.2\%$. This trade-off is consistent across different LLMs and aligns with our theoretical analysis in Theorem~\ref{theorem:bounds}. To better balance the utility-detectability trade-off, we set $\delta = 0.92$ in our experiments.

\begin{table}[t]
\footnotesize
    \centering
    \caption{Evaluation criteria for text selection in AMT experiments.}
    % \footnotesize % Adjust font size
    \setlength{\arrayrulewidth}{0.95pt} % Adjust border thickness
    \begin{tabular}{ |p{1.5cm}|p{3.8cm}|p{7.0cm}|}
        \hline
        \textbf{Items} & \textbf{Descriptions} & \textbf{Example} \\
        \hline
        \makecell[l]{ \textbf{Relevance to} \\ \textbf{Original } \\ \textbf{Output}} & 
        \makecell[l]{ Choose the one that accurately \\ retains the original meaning. \\ Avoid unnecessary changes that \\ alter the intent.} & 
        \makecell[l]{ 
            \textbf{Original:} She \textcolor{green}{has} 3 apples! \\
            \textbf{Choice A:} She \textcolor{blue}{possesses} 3 apples! \textcolor{green}{\ding{51}} (Accurate, slight wording \\ change) \\
            \textbf{Choice B:} She \textcolor{red}{receives} 3 apples! \textcolor{red}{\ding{55}} (Implies a different action, \\ not ownership) \\
            \textbf{Choice A is better. Answer: A}
        } \\
        \hline
    \makecell[l]{\textbf{Grammatical} \\ \textbf{Correctness}} & 
        \makecell[l]{Select the sentence that follows \\ correct grammar including \\verb-noun agreement, tense,\\ punctuation, etc.  Ensure names, \\ places, and organizations remain \\correctly capitalized.} & 
        \makecell[l]{
            \textbf{Original:} The \textcolor{green}{operation-based system} was better.\\
            \textbf{Choice A:} The operation-based \textcolor{blue}{infrastructure} was better. \textcolor{green}{\ding{51}} \\
            \textbf{Choice B:} The \textcolor{red}{operation - based cluster were} better. \textcolor{red}{\ding{55}} \\ (Wrong subject-verb agreement) \\
            \textbf{Choice A is better. Answer: A}
        } \\
        \hline
       \makecell[l]{\textbf{Factual} \\ \textbf{Accuracy}} & 
        \makecell[l]{Ensure that dates, numbers, \\ locations, names, and factual \\information remain unchanged. Do \\ not select a choice that modifies \\critical data.} & 
        \makecell[l]{
            \textbf{Original:} It was at \textcolor{green}{1:30 PM to 3:30 PM} on Monday.\\
            \textbf{Choice A:} It was at \textcolor{red}{1:33 PM to 4:00 PM} on Monday. \textcolor{red}{\ding{55}} (Wrong \\ information)\\
            \textbf{Choice B:} It \textcolor{blue}{happened} 1:30 PM to 3:30 PM on Monday. \textcolor{green}{\ding{51}} \\
            \textbf{Choice B is better. Answer: B}
        } \\
        \hline
      \makecell[l]{\textbf{Sentence}\\ \textbf{Structure} \\ \textbf{and Style}}   & 
        \makecell[l]{Maintain a structure that is similar \\to the original, ensuring clarity and \\fluency. Avoid overly complicated\\ or unnatural phrasing.} & 
        \makecell[l]{
            \textbf{Original:} She \textcolor{green}{was} excited to visit New York last summer.\\
            \textbf{Choice A:} She \textcolor{blue}{felt} excited to visit New York last summer. \textcolor{green}{\ding{51}} \\
            \textbf{Choice B:} \textcolor{red}{New York was the place she last summer went} \\ \textcolor{red}{excited to.} \textcolor{red}{\ding{55}} (Awkward structure, unclear meaning) \\
            \textbf{Choice A is better. Answer: A}
        }  \\
        \hline
    \end{tabular}
    \label{tab:evaluation_criteria} 
\end{table}

\begin{figure}[t]
\centering
% \subfigure[LLaMA-2]
\includegraphics[scale=0.18]{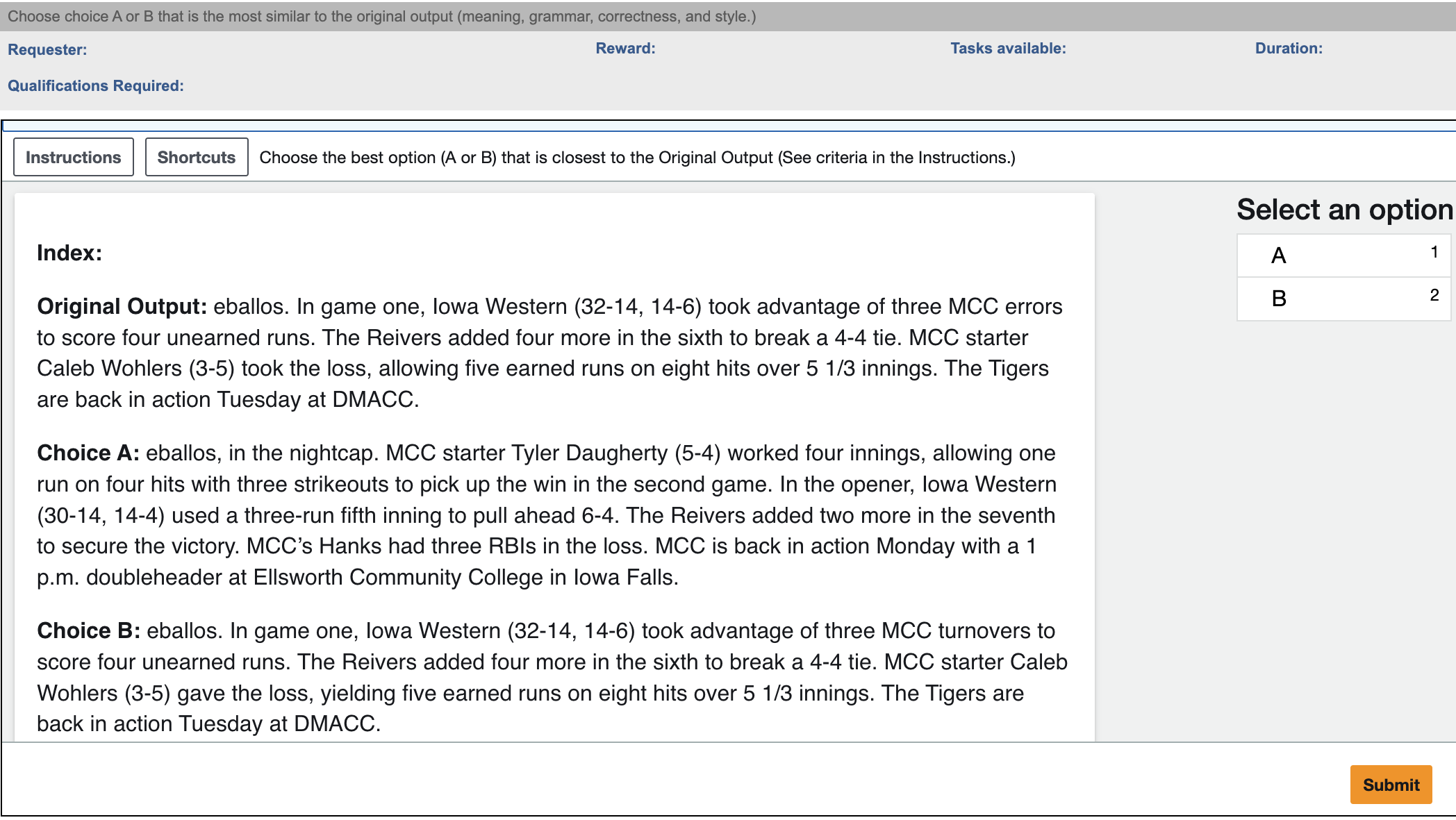}\hfill
% \subfigure[Mistral]{\includegraphics[scale=0.22]{samples/images/WM_RemovalMistral.png}}\vfill
\caption{ATM template for an example evaluation.} %\vspace{-7.5pt}
\label{fig:AMT_1}
\end{figure}  

\subsection{Human Evaluation Instruction} 
For each AMT task, we provide comprehensive and step-by-step instructions to guide   workers in assessing and  comparing watermarked outputs from different watermarking techniques.  These instructions are carefully designed to standardize the evaluation process and minimize subjectivity. Workers select the best watermarked output based on the following four criteria:  (1) relevance to the original LLM output, assessing how well the watermarked version retains the original semantic meaning; (2) grammatical correctness, evaluating fluency and adherence to standard grammar rules; (3) factual consistency, ensuring that the content remains and align well with the original output; and (4) preservation of sentence structure and stylistic tone, which examines whether the stylistic and structural qualities of the original text are maintained. These evaluation criteria are selected to ensure that the chosen output not only preserves the semantic content but also maintains high linguistic and stylistic quality.
 A detailed example of the evaluation instructions and rubrics for the head-to-head evaluation in text generation task is provided in Table~\ref{tab:evaluation_criteria}. 
 In addition, a snapshot of the AMT task template is shown in Fig.~\ref{fig:AMT_1}, which outlines the specific interface used by the workers during the evaluation process, ensuring consistency and clarity in their assessments.

\clearpage
\newpage

\end{document}